\documentclass[twoside,leqno]{article}

\usepackage{ifthen}
\usepackage{comment}
\usepackage[
style=numeric,
giveninits=true, 
maxbibnames=99, 
maxcitenames=2, 
natbib=true, 
labelnumber, 
backend=biber, 
sortcites=true,
doi=false,isbn=false,url=false
]{biblatex}
\renewbibmacro{in:}{%
  \ifentrytype{article}{}{\printtext{\bibstring{in}\intitlepunct}}
}
\let\oldcitet\citet
\renewcommand{\citet}[1]{\mbox{\oldcitet{#1}}}

\usepackage[colorlinks=true,bookmarks=false]{hyperref}
\usepackage[small,bf]{caption}
\usepackage{graphicx}
\usepackage{subcaption}
\usepackage[british]{babel}
\usepackage{paralist}
\usepackage[shortlabels]{enumitem}
\setlist[enumerate]{nosep} %
\setlist[itemize]{nosep} %

\usepackage{xcolor}
\definecolor{darkblue}{rgb}{0,0,0.45}
\definecolor{darkred}{rgb}{0.6,0,0}
\definecolor{darkgreen}{rgb}{0.13,0.5,0}
\hypersetup{colorlinks, linkcolor=darkblue, citecolor=darkgreen,
urlcolor=darkblue}

\usepackage{amsmath,amsfonts,amssymb}
\usepackage{mathtools}
\usepackage{mathrsfs}
\usepackage{thm-restate}
\usepackage{wrapfig}
\usepackage{xspace}
\usepackage{csquotes}
\usepackage{multirow}

\usepackage[nameinlink,capitalise]{cleveref}

\usepackage{footmisc}

\crefname{theorem2}{Theorem}{Theorems}
\newtheorem{claim}{Claim}[section]
\newtheorem{prop}{Proposition}[section]

\newcommand{\eps}{{\varepsilon}}

\newcommand{\hy}{\hbox{-}\nobreak\hskip0pt}

\usepackage{environ,xstring}
\newif\iflabel
\newif\ifdbs
\newif\ifamp
\NewEnviron{doitall}{%
  \noexpandarg
  \expandafter\IfSubStr\expandafter{\BODY}{\label}{\labeltrue}{\labelfalse}%
  \expandafter\IfSubStr\expandafter{\BODY}{\\}{\dbstrue}{\dbsfalse}%
  \expandafter\IfSubStr\expandafter{\BODY}{&}{\amptrue}{\ampfalse}%
  \iflabel\def\doitallstar{}\else\def\doitallstar{*}\fi
  \ifdbs
    \ifamp
      \def\doitallname{align}%
    \else
      \def\doitallname{multline}%
    \fi
  \else
    \def\doitallname{equation}
  \fi
  \begingroup\edef\x{\endgroup
    \noexpand\begin{\doitallname\doitallstar}%
    \noexpand\BODY
    \noexpand\end{\doitallname\doitallstar}%
  }\x
}
\def\[#1\]{\begin{doitall}#1\end{doitall}}

\newcommand{\pname}[1]{\textsc{#1}}

\usepackage{framed}
\newcommand{\executeiffilenewer}[3]{%
\ifnum\pdfstrcmp{\pdffilemoddate{#1}}%
{\pdffilemoddate{#2}}>0%
{\immediate\write18{#3}}\fi%
}

\usepackage{ltexpprt}

\bibliography{papers}

\newcommand{\polyn}{\cdot n^{O(1)}}
\DeclareMathOperator{\cost}{cost}

\newcommand{\SNDP}{\pname{SNDP}\xspace}
\newcommand{\ECSNDP}{\pname{EC\hy{}SNDP}\xspace}
\newcommand{\LCSNDP}{\pname{LC\hy{}SNDP}\xspace}
\newcommand{\VCSNDP}{\pname{VC\hy{}SNDP}\xspace}
\newcommand{\ec}{\ensuremath{\lambda}}
\newcommand{\lc}{\ensuremath{\kappa'}}
\newcommand{\vc}{\ensuremath{\kappa}}
\newcommand{\solsize}{\ensuremath{\ell}}
\newcommand{\nrterm}{\ensuremath{k}}
\newcommand{\sumdem}{\ensuremath{D}}
\newcommand{\maxdem}{\ensuremath{d_{\max}}}
\newcommand{\dem}{\ensuremath{d}}
\newcommand{\terms}{\ensuremath{R}}
\newcommand{\tw}{\ensuremath{tw}}
\newcommand{\DST}{\pname{\ensuremath{2}\hy{}DST}\xspace}

\title{\Large The Parameterized Complexity of\\
the Survivable Network Design Problem\footnote{An extended abstract of this 
paper appeared at the 5th Symposium on Simplicity in Algorithms (SOSA@SODA 
2022)~\cite{DBLP:conf/sosa/Feldmann0L22}}}

\author{Andreas Emil Feldmann\footnote{
Supported by the Czech Science Foundation GA{\v C}R (grant \#19-27871X).
}
\thanks{Charles University, Prague, Czechia}
\and
Anish Mukherjee\footnote{Supported by the ERC CoG grant TUgbOAT no 772346.}
\thanks{University of Warsaw, Warsaw, Poland}
\and
Erik Jan van Leeuwen
\thanks{Utrecht University, Utrecht, The Netherlands}
}

\date{}

\begin{document}

\maketitle

\begin{abstract}
For the well-known \pname{Survivable Network Design Problem (SNDP)} we are 
given an undirected graph~$G$ with edge costs, a set $\terms$ of terminal 
vertices, and an integer demand $\dem_{s,t}$ for every terminal pair $s,t\in 
\terms$. The task is to compute a subgraph $H$ of $G$ of minimum cost, such that 
there are at least $\dem_{s,t}$ disjoint paths between $s$ and $t$ in $H$. 
Depending on the type of disjointness we obtain several variants of SNDP that 
have been widely studied in the literature: if the paths are required to be 
edge-disjoint we obtain the edge-connectivity variant~(\ECSNDP), while 
internally vertex-disjoint paths result in the vertex-connectivity variant 
(\VCSNDP). Another important case is the element-connectivity variant (\LCSNDP), 
where the paths are disjoint on edges and non-terminals, i.e., they may only 
share terminals. 

In this work we shed light on the parameterized complexity of the above 
problems. We consider several natural parameters, which include the solution 
size~$\solsize$, the sum of demands~$\sumdem$, the number of 
terminals~$\nrterm$, and the maximum demand~$\maxdem$. Using simple, elegant 
arguments, we prove the following results.
\begin{itemize}
\item We give a complete picture of the parameterized tractability of the three 
variants w.r.t.\ parameter~$\solsize$: both \ECSNDP and \LCSNDP are FPT, while 
\VCSNDP is W[1]-hard (even in the uniform single-source case with~$\nrterm=3$).
\item We identify some special cases of \VCSNDP that are FPT:
\begin{itemize}
 \item when $\maxdem\leq 3$ for parameter $\solsize$,
 \item on locally bounded treewidth graphs (e.g., planar graphs) for parameter 
$\solsize$, and
 \item on graphs of treewidth $\tw$ for parameter $\tw+\sumdem$, which is in 
contrast to a result by Bateni et al.~[JACM~2011] who show NP-hardness for 
$\tw=3$ (even if $\maxdem=1$, i.e., the \pname{Steiner Forest} problem).
\end{itemize}
\item The well-known \pname{Directed Steiner Tree (DST)} problem can be seen as 
single-source \ECSNDP with $\maxdem=1$ on directed graphs, and is FPT 
parameterized by $\nrterm$ [Dreyfus \& Wagner 1971]. We show that in contrast, 
the 2-DST problem, where $\maxdem=2$, is W[1]\hy{}hard, even when 
parameterized by~$\solsize$ (which is always larger than $\nrterm$).
\end{itemize}
\end{abstract}

\section{Introduction}

Network design is an algorithmic research area that investigates 
connecting a set of nodes in a network in the cheapest possible way. A 
well-known, classical example is the \pname{Steiner Tree} problem, where we are 
given an undirected graph $G=(V,E)$ with non-negative edge costs and a terminal 
set $\terms\subseteq V$. The aim is to find the cheapest tree in $G$ containing 
all terminals of~$\terms$. This is one of the first NP-hard problems given by 
\citet{karp1975computational} in~1975. The \pname{Steiner Tree} problem finds 
many applications (see surveys~\cite{cheng2013steiner, voss2006steiner, 
tang2020survey}), and can for instance model a scenario in which several nodes 
in a telecommunication network (e.g., the Internet) should form a subnetwork 
while minimizing the cost of paying for the involved connections 
(cf.~\cite{voss2006steiner}).

Many variations of \pname{Steiner Tree} have been studied in this active 
research area (see surveys \cite{hwang1992steiner, du2013advances, 
gupta2011approximation}). One well-known generalization is obtained by 
introducing a redundancy requirement so that terminals are connected not only 
with one path but several, in order to guarantee connectivity even if some parts 
of the network should fail. This leads to the so-called \pname{Survivable 
Network Design Problem~(\SNDP)}, which is the main focus of this paper and is 
central to the area of network design (see 
surveys~\cite{kortsarz2007approximating, kerivin2005design, 
soni1999survivable}). Formally, we are given an undirected graph $G=(V,E)$ 
together with a non-negative edge cost function denoted by 
$\cost:E\to\mathbb{R}^+$, a terminal set $\terms\subseteq V$, and a non-negative 
integer demand $\dem_{s,t}\in\mathbb{N}_0$ for each terminal pair 
$s,t\in\terms$. The aim is to find a subgraph~$H\subseteq G$ containing $\terms$ 
so that every terminal pair $s,t\in\terms$ is connected by at least~$\dem_{s,t}$ 
disjoint paths in~$H$, while minimizing $\cost(H)=\sum_{e\in E(H)}\cost(e)$. 
Depending on the type of path-disjointness we obtain several variants of this 
problem. In particular, the \emph{edge-connectivity}~$\ec_H(s,t)$ between two 
terminals $s$ and $t$ in $H$ is the maximum number of edge-disjoint paths 
between them, while the \emph{vertex-connectivity} $\vc_H(s,t)$ is the maximum 
number of internally vertex-disjoint paths. If the aim is to compute a solution 
$H$ such that $\ec_H(s,t)\geq\dem_{s,t}$ for every terminal pair $s,t\in\terms$, 
we obtain the edge-connectivity variant called \ECSNDP, and if the constraint is 
replaced by $\vc_H(s,t)\geq\dem_{s,t}$ we have the vertex-connectivity variant 
denoted by \VCSNDP. Note that if $\dem_{s,t}=1$ for all $s,t\in\terms$, each of 
these problems becomes the \pname{Steiner Tree} problem, while if 
$\dem_{s,t}\in\{0,1\}$ we obtain the so-called \pname{Steiner Forest} problem 
(since now an optimum solution will be a forest).

Because the \pname{Steiner Tree} problem is NP-hard~\cite{karp1975computational} 
we do not expect any polynomial-time algorithms to solve \SNDP. One approach to 
get around this issue that has been thoroughly studied for \SNDP, is to find 
polynomial-time \emph{approximation algorithms}~\cite{vazirani2013approximation, 
williamson2011design}: for some $\alpha>1$, such an algorithm computes an 
\emph{$\alpha$-approximate} solution, which costs at most $\alpha$ times more 
than the optimum. \citet{jain2001factor} showed that \ECSNDP has a 
polynomial-time 2-approximation algorithm. In this seminal work he introduced 
the iterative rounding technique for linear programs, which is considered a 
milestone in the development of approximation algorithms and has since then been 
applied to many different problems~\cite{lau2011iterative}. In contrast to 
\ECSNDP, it is known that \VCSNDP is much harder to approximate: given a graph 
with $n$ vertices, when $\maxdem=\max\{\dem_{s,t}\mid s,t\in\terms\}$ is 
polynomial in~$n$, there is no 
$(2^{\log^{1-\eps}n})$\hy{}approximation~\cite{kortsarz2004hardness} for any 
$\eps>0$, unless NP admits quasi-polynomial-time algorithms. For constant values 
of $\maxdem$, no $\maxdem^\eps$-approximation algorithm 
exists~\cite{chakraborty2008network} given that $\maxdem\geq d_0$ for some 
universal constants $d_0$ and~$\eps$, unless P=NP. On the positive side, if 
$\nrterm=|\terms|$, \citet{chuzhoy2009k} show that an 
$O(\maxdem^3\log\nrterm)$\hy{}approximation can be computed in polynomial time. 
These results reflect a typical behaviour of network design (and other) 
problems, namely that the ``vertex-version'' of a problem (in this case~\VCSNDP) 
is usually computationally harder than its ``edge-version'' (in this 
case~\ECSNDP).

The algorithm by~\citet{chuzhoy2009k} for \VCSNDP exploits known results for 
another important variant of \SNDP, for which the paths connecting terminals are 
supposed to be \emph{element-disjoint}, where the ``elements'' are the 
non-terminals (also called \emph{Steiner vertices}) and edges. That is, the 
\emph{element-connectivity} between $s$ and~$t$ in~$H$, denoted by~$\lc_H(s,t)$, 
is the maximum number of paths from $s$ to $t$ that may only share terminals, 
and for the element-connectivity variant called \LCSNDP the aim is to compute a 
minimum cost solution~$H$ for which $\lc_H(s,t)\geq\dem_{s,t}$ for 
all~$s,t\in\terms$. While at first glance the element-connectivity seems more 
akin to the vertex-connectivity, surprisingly the iterative rounding 
2-approximation algorithm for \ECSNDP can be 
generalized~\cite{fleischer2006iterative} to \LCSNDP, making this problem 
computationally similar to \ECSNDP instead of~\VCSNDP.

In this paper we shed new light on these problems from the point-of-view of 
parameterized complexity~\cite{cygan2015parameterized}, which is a different 
popular approach to obtain efficient algorithms for NP-hard problems: we are 
given a parameter~$p\in\mathbb{N}_0$ and the aim is to compute an optimum 
solution in~$f(p)\polyn$ time for some computable function~$f$. If such an 
algorithm exists the problem is called \emph{fixed-parameter tractable (FPT)} 
for parameter $p$, and the algorithm is correspondingly called an \emph{FPT 
algorithm}. While for NP-hard problems we expect this function~$f$ to be 
super-polynomial (e.g.,~exponential), the rationale is that for applications in 
which the parameter is small such an algorithm solves the problem efficiently.

Some special cases of \SNDP have been considered from the parameterized 
perspective before. Most prominently, the classic Dreyfus and 
Wagner~\cite{dreyfus1971steiner} algorithm can be used to solve the 
\pname{Steiner Tree} problem in $3^\nrterm\polyn$ time. This runtime was later 
improved~\cite{fuchs2007dynamic} to $(2+\eps)^\nrterm\cdot 
n^{O(\sqrt{1/\eps}\log(1/\eps))}$ for any $\eps>0$, and to $2^\nrterm\polyn$ if 
the edge weights are polynomially 
bounded~\cite{nederlof2009fast,bjorklund2007fourier}. Moreover, in the special 
case of \ECSNDP when $V=\terms$ and all vertex pairs $s,t\in V$ have uniform 
demand $\dem_{s,t}=\dem$ for some given~$\dem$, we obtain the \pname{$\dem$-Edge 
Connected Subgraph} problem. \citet{bang2018parameterized} show that this 
problem is FPT for the combined parameter of $\dem$ and the size of a deletion 
set, i.e., the number of edges to be removed from the input graph in order to 
obtain a minimum cost solution. For the same parametrization, 
\citet{gutin2019path} provide a (non-uniform) FPT algorithm for the 
vertex-connectivity variant called \pname{$\dem$-Vertex Connected Subgraph} on 
unweighted graphs. The authors of~\cite{bang2018parameterized} also note that 
requiring a spanning solution $H$ (i.e., when $\dem_{s,t}\geq 1$ for all $s,t\in 
V$) makes \SNDP trivially FPT when parameterizing by the solution size (i.e, the 
number of edges of the solution~$H$), since any spanning subgraph has at 
least~$n-1$ edges. On the other hand, for non-spanning solutions 
\citet{bang2018parameterized} prove that the parameterization by the deletion 
set already renders the \pname{Steiner Forest} problem W[1]-hard.

In this paper we consider the above three variants of \SNDP in their full 
generality, and we obtain a complete classification for the natural 
parameterization by the solution size. As summarized by the following two 
theorems, our results reflect the previous findings for approximation algorithms 
that \ECSNDP and \LCSNDP have similar complexities and are computationally 
easier than \VCSNDP.

\begin{restatable}{theorem2}{LCalg}
\label{thm:LC-alg}
Both \ECSNDP and \LCSNDP can be solved in $2^{O(\solsize\log\solsize)}\polyn$ 
time, where $\solsize$ is the number of edges of the solution and $n$ is the 
number of vertices of the input graph.
\end{restatable}

In contrast, the following result shows that \VCSNDP is not FPT for 
parameter~$\solsize$, unless FPT=W[1]. This hardness result is true even in the 
\emph{single-source} case where we have a fixed terminal $r\in R$ called the 
\emph{root} and any terminal only needs to be connected to the root, i.e., a 
demand $\dem_{s,t}$ is positive only if $r\in\{s,t\}$, and 
otherwise~$\dem_{s,t}=0$. It is known that single-source \VCSNDP has a 
polynomial-time $O(\maxdem^2)$\hy{}approx\-imation 
algorithm~\mbox{\cite{nutov2012approximating, nutov2018erratum}}, improving 
over the known approximation for the general case. An even better 
$O(\dem\log\dem)$\hy{}approximation can be 
computed~\mbox{\cite{nutov2012approximating, nutov2018erratum}} for 
\emph{uniform} single-source \VCSNDP, meaning that $\dem_{r,t}=\dem$ for a 
given~$\dem$ and all terminals~$t\in\terms$. Furthermore, note that if there are 
only two terminals then a simple min-cost flow computation will solve \VCSNDP in 
polynomial-time. We show that already increasing the number of terminals by one 
makes the problem hard.

\begin{restatable}{theorem2}{VChard}
\label{thm:VC-hard}
Uniform single-source \VCSNDP is \textup{W[1]}-hard parameterized by the 
number~$\solsize$ of edges of the solution, even if the number $k$ of terminals 
is $3$.
\end{restatable}

We remark that the previously mentioned approximation lower bound 
of~\cite{kortsarz2004hardness} uses a reduction that together with the results 
of~\cite{dinur2018eth,manurangsi2021strongish} implies that \VCSNDP 
parameterized by the number~$\nrterm$ of terminals has no $f(\nrterm)\polyn$ 
time algorithm for any function~$f$ that computes a 
$\nrterm^{1/4-o(1)}$\hy{}approximation under Gap-ETH, or computes a 
$\nrterm^{1/2-o(1)}$\hy{}approximation under the Strongish Planted Clique 
Hypothesis. \cref{thm:VC-hard} nicely complements these hardness results, since 
the maximum demand~$\maxdem$ is unbounded in the reduction given 
in~\cite{kortsarz2004hardness}, and thus does not provide hardness 
parameterized by the solution size~$\solsize$, as $\solsize\geq\maxdem$.

In light of \cref{thm:VC-hard} we identify several special cases of \VCSNDP that 
are FPT. When $\maxdem\leq 2$, \citet{fleischer20012} showed that the iterative 
rounding 2-approximation algorithms for \ECSNDP and \LCSNDP can be generalized 
to \VCSNDP. On the other hand, she also showed that this approach cannot be 
generalized to~$\maxdem=3$.  Interestingly, we prove that our techniques to 
obtain the FPT algorithms of \cref{thm:LC-alg} can be generalized to 
$\maxdem\leq 3$ (but no further; see \cref{sec:techniques}).

\begin{restatable}{theorem2}{VCalg}
\label{thm:VC-alg}
\VCSNDP can be solved in $2^{O(\solsize\log\solsize)}\polyn$ time if 
$\maxdem\leq 3$, where $\solsize$ is the number of edges of the solution and $n$ 
is the number of vertices of the input graph.
\end{restatable}

An important graph class that has been thoroughly studied in network design (and 
elsewhere) are planar graphs. For instance \citet{borradaile2009n} show that on 
such graphs a polynomial-time approximation scheme (PTAS) exists for the 
\pname{Steiner Tree} problem, and also for the \pname{Steiner Forest} problem as 
shown by \citet{bateni2011approximation}. For \SNDP on the other hand, to the 
best of our knowledge it is not known whether better approximations can be 
computed for planar graphs than in general. We prove that \VCSNDP is FPT for 
parameter~$\solsize$ on planar input graphs. In fact our result holds for the 
more general class of locally bounded treewidth graphs, for which any subgraph 
has treewidth bounded as a function of its diameter (see \cref{sec:tw} for 
formal definitions). This also means that our result extends to apex-minor-free 
graphs.

\begin{restatable}{theorem2}{ltwalg}
\label{thm:ltw-alg}
\VCSNDP on graphs of locally bounded treewidth (e.g., planar and 
apex-minor-free graphs) can be solved in $f(\solsize)\cdot n$ time for some 
function $f$, where $\solsize$ is the number of edges of the solution and $n$ is 
the number of vertices of the input graph.
\end{restatable}

To obtain the algorithm of \cref{thm:ltw-alg} we rely on an FPT algorithm 
parameterized by the treewidth (i.e., not locally bounded). While the 
\pname{Steiner Tree} problem is solvable~\cite{bodlaender2015deterministic} in 
$2^{O(\tw)}\cdot n$ time on graphs of treewidth~$\tw$, we cannot hope for an FPT 
algorithm for \SNDP when using only the treewidth as a parameter, since it was 
shown by \citet{bateni2011approximation} that the special case of \pname{Steiner 
Forest} is NP-hard for graphs of treewidth~3. Instead we combine the parameter 
with the sum of demands $\sumdem=\sum_{s,t\in R}\dem_{s,t}$ to obtain the 
following result. Note that~$\sumdem$ can be linearly upper bounded by 
$\solsize$ and thus $\sumdem$ is the stronger parameter, which means that 
\cref{thm:VC-hard} excludes an FPT algorithm using only~$\sumdem$.

\begin{restatable}{theorem2}{twDalg}
\label{thm:twD-alg}
\VCSNDP can be solved in $f(\tw+\sumdem)\cdot n$ time for some function $f$, 
where $\tw$ is the treewidth and $n$ the number of vertices of the input graph, 
and $\sumdem$ is the sum of demands.
\end{restatable}

When considering directed graphs as inputs, several well-studied problems can 
be seen as special cases of \SNDP, and are typically much harder than their 
undirected counterparts. For instance, already the special case when 
$\maxdem=1$, i.e., the \pname{Directed Steiner Network (DSN)} problem, is very 
hard as shown by \citet{dinur2018eth}: no $\nrterm^{1/4-o(1)}$\hy{}approximation 
algorithm with runtime $f(\nrterm)\polyn$ exists for any computable 
function~$f$, under Gap-ETH. More recently, under the Strongish Planted Clique 
Hypothesis, \citet{manurangsi2021strongish} showed that no 
$\nrterm^{1/2-o(1)}$\hy{}approximation can be computed with this 
runtime.\footnote{These hardness results also imply the above mentioned 
approximation hardness for \VCSNDP, via a reduction from \pname{DSN}.} 
This result is also asymptotically tight, since an 
$O(\nrterm^{1/2+\eps})$-approximation can be 
computed~\cite{chekuri2011set,feldman2012improved} in polynomial time for 
any~$\eps>0$. \citet{feldmann2016complexity} gave a complete characterization of 
which special cases of DSN are FPT and which are W[1]-hard parameterized 
by~$\nrterm$, depending on the structure of the non-zero demands. For instance, 
the single-source case with $\maxdem=1$, which is called the \pname{Directed 
Steiner Tree~(DST)} problem (since an optimum solution is an arborescence), is 
FPT~\cite{dreyfus1971steiner, feldmann2016complexity}. However, in polynomial 
time no $O(\log^{2-\eps}\nrterm)$\hy{}approx\-imation can be computed for 
DST~\cite{halperin2003polylogarithmic}, unless NP admits randomized 
quasi-polynomial-time algorithms. The currently best polynomial-time 
approximation algorithm~\cite{charikar1999approximation} achieves a ratio 
of~$O(\nrterm^\eps)$ for any constant $\eps>0$. The generalization known as 
$\dem$-DST is the directed edge-connectivity variant with uniform single-source 
demands~$\dem$. For \emph{quasi-bipartite} input graphs, where Steiner vertices 
form an independent set, an $O(\log\nrterm\log\dem)$\hy{}approx\-imation 
algorithm was obtained for $\dem$-DST by~\citet{chan2020polylogarithmic}. For 
2-DST on general input graphs, \citet{grandoni2017surviving} give an 
$O(\nrterm^\eps\log n)$-approximation algorithm for any~$\eps>0$.

While \pname{DST} is FPT~\cite{dreyfus1971steiner, feldmann2016complexity} 
parameterized by~$\nrterm$, we show that increasing the demand by only one 
compared to DST (i.e., the 2-DST problem) already yields W[1]\hy{}hardness, even 
when using the much weaker parameter $\solsize$ instead of $\nrterm$. This is in 
stark contrast to undirected graphs, as witnessed by 
\cref{thm:LC-alg,thm:VC-alg}.

\begin{restatable}{theorem2}{DSThard}
\label{thm:2DST-hard}
The \pname{$2$-DST} problem is \textup{W[1]}-hard parameterized by the number 
$\solsize$ of edges of the solution.
\end{restatable}

\subsection{Our techniques.}\label{sec:techniques}

\begin{wrapfigure}[18]{r}{0.25\textwidth}
  \vspace{-6mm}
  \begin{center}
    \includegraphics[width=0.2\textwidth]{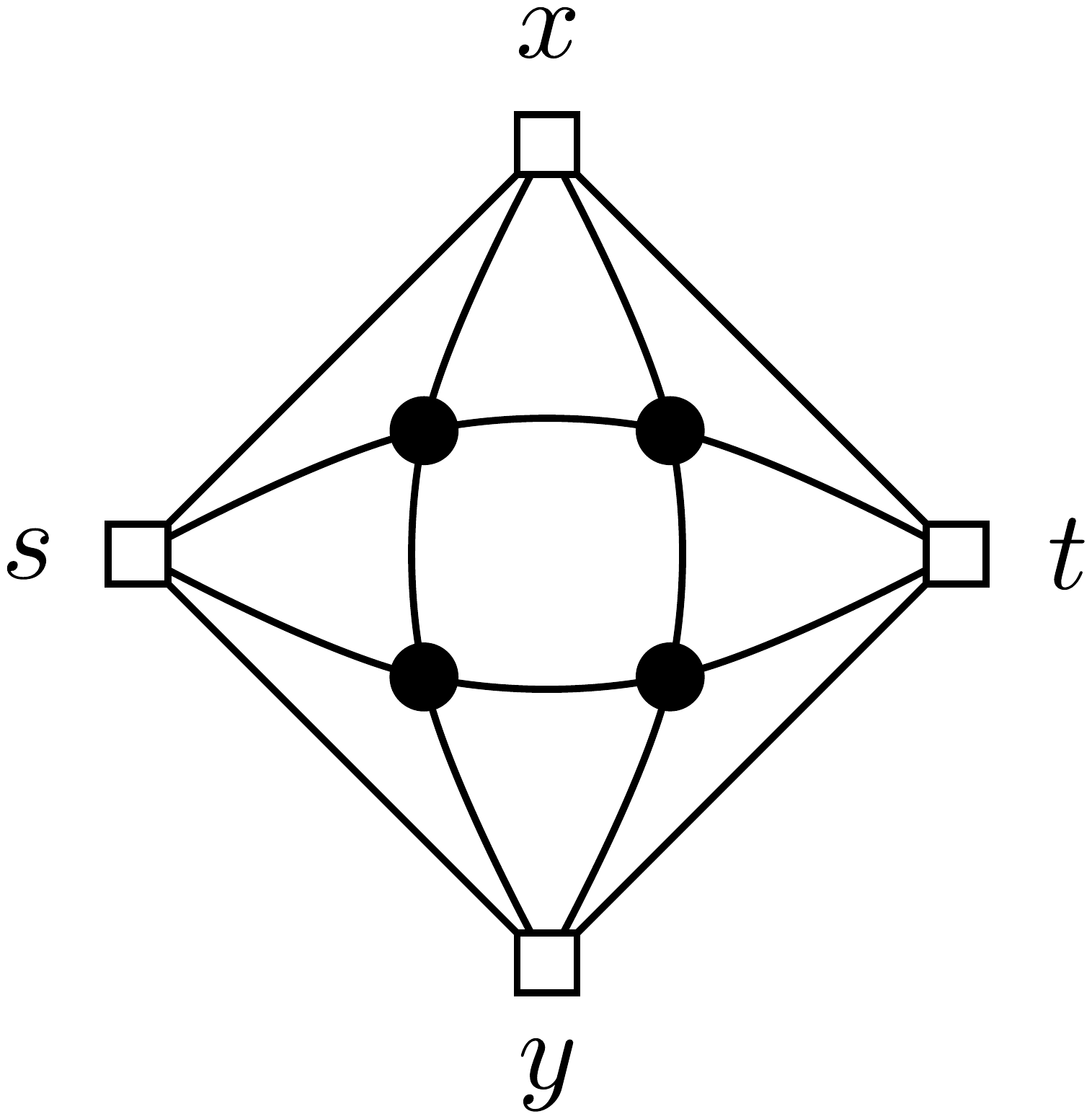}
  \end{center}
  \caption{If the demand between every terminal pair (white boxes) is~$4$, then 
the only \VCSNDP solution in the shown graph is the whole graph, which has a 
cycle on the Steiner vertices (black dots).}
  \label{fig:dem4}
\end{wrapfigure}
We use simple, elegant techniques to prove our theorems. The main ingredient 
for the FPT algorithms for \ECSNDP and \LCSNDP of \cref{thm:LC-alg} is the 
``reduction lemma'' of \citet{chekuri2014graph}. They prove that it is possible 
to either delete or contract any edge between two Steiner vertices without 
reducing the element-connectivity between any terminal pair. As we prove in 
\cref{sec:LC-alg}, this implies that an optimum solution to \LCSNDP has no 
cycles on the Steiner vertices and can thus be decomposed into internally 
vertex-disjoint trees for which the leaves are the terminals. We then use the 
well-known colour coding technique~\cite{alon1995color} to colour all vertices 
of the input graph in such a way that every tree of the decomposition of the 
optimum solution is coloured by a unique colour. Finally, we invoke known FPT 
algorithms~\cite{dreyfus1971steiner} for the \pname{Steiner Tree} problem in 
order to compute each tree of the decomposition of the \LCSNDP solution 
independently. To solve \ECSNDP we show that our structural insights on 
solutions to \LCSNDP implies that optimum solutions to \ECSNDP can be decomposed 
into edge-disjoint trees, via a simple reduction. We can then again use colour 
coding, where now we colour terminals and edges, to obtain a very similar 
algorithm to before.

Also our FPT algorithm of \cref{thm:VC-alg} for \VCSNDP with $\maxdem\leq 3$ 
uses the same approach. However, for this we need to extend the reduction lemma 
of \citet{chekuri2014graph} to vertex-connectivity. We do this by adding some 
observations to their proof in \cref{sec:VC-alg}, to show that when deleting or 
contracting an edge the vertex-connectivity between any terminal pair can only 
decrease if their connectivity is at least~4. As a consequence, any optimum 
solution to \VCSNDP can be decomposed into trees if the demands are at most~3. 
In fact this is best possible, since there is an example with $\maxdem=4$ for 
which the optimum solution cannot be decomposed into trees in this way, since 
there exists a cycle among the Steiner vertices (cf.~\cref{fig:dem4}).

The algorithm of \cref{thm:twD-alg} for the parameterization by the treewidth 
and sum of demands can be obtained via a straight-forward dynamic program on the 
tree decomposition of the graph. In the interest of simplicity, however, in this 
paper we do not present this dynamic program directly, as it would take a 
lengthy formal argument that would not add much insight into the problem. 
Instead, in \cref{sec:tw} we express the problem using an MSOL formula and then 
invoke a very general Courcelle-type theorem (which at its heart also involves a 
dynamic program). This means that the runtime we obtain is far from optimal 
compared to a direct formulation of the needed dynamic program, but the proof 
that there is an FPT algorithm for parameter $\tw+\sumdem$ is very simple and 
short. Given this algorithm it is then easy to obtain the FPT algorithm of 
\cref{thm:ltw-alg} for graphs of locally bounded treewidth parameterized by the 
solution size, as also detailed in \cref{sec:tw}.

The hardness result of \cref{thm:2DST-hard} for the 2-DST problem is a rather 
straightforward reduction from the W[1]-hard \pname{Multicoloured Clique} 
problem and can be found in \cref{sec:2dst}. To prove hardness of \VCSNDP in 
\cref{thm:VC-hard} however, a standard reduction from \pname{Multicoloured 
Clique} seems hard to obtain due to the undirectedness of the input graphs: it 
is not clear how to prevent paths from passing through the wrong gadgets. To 
overcome this obstacle, in \cref{sec:ss} we reduce from the W[1]-hard 
\pname{Grid Tiling} problem instead, and exploit its grid structure in order to 
control the routes taken by the paths.

Finally, in \cref{sec:open} we list some open problems.

\subsection{Additional related results.}\label{sec:related}

As can be seen from above, a vast literature exists on \SNDP and related 
problems. We add just a few more closely related results here.

In addition to undirected graphs, \citet{bang2018parameterized} also show that 
in directed input graphs both \pname{$\dem$-Edge} and \pname{$\dem$-Vertex 
Connected Subgraph} are FPT parameterized by $\dem$ and the size of a deletion 
set. While a brute-force algorithm for these problems can find an optimum 
solution in~$2^{O(\dem n(\log\dem +\log n))}$ time, for \pname{$\dem$-Edge 
Connected Subgraph} this was improved to a single-exponential $2^{O(\dem n)}$ 
runtime by~\citet{agrawal2017fast} for both directed and undirected graphs. 

The \pname{Steiner Tree} problem was shown to admit a 
$(\ln(4)+\eps)$-approximation algorithm in the seminal work of 
\citet{byrka2013steiner}. It is also known that this problem is 
APX-hard~\cite{chlebik2008steiner}. Using an FPT 
algorithm~\cite{dreyfus1971steiner} for \pname{Steiner Tree} it is not hard to 
prove that also \pname{Steiner Forest} is FPT for the number of terminals (even 
in single-exponential time; cf.~\cite{chitnis2021parameterized}). For the 
parameterization by the number of Steiner vertices in the optimum solution, a 
folklore result says that \pname{Steiner Tree} is W[2]-hard 
(cf.~\cite{cygan2015parameterized, dvorak2021parameterized}). However a 
parameterized approximation scheme exists for this 
parameter~\cite{dvorak2021parameterized}. Similar results have been found for 
special cases of DSN~\cite{chitnis2021parameterized}.

A useful and well-known graph operation introduced by \citet{lovasz1976some} is 
called \emph{splitting-off} and entails replacing two edges $uv$ and $wv$ 
incident to a vertex $v$ by a direct edge $uw$. This operation can be performed 
in such a way that the edge-connectivity of the graph remains unchanged, i.e., 
it preserves the global edge-connectivity. \citet{mader1978reduction} 
generalized this to preserve local edge-connectivity, i.e., the 
edge-connectivity between every vertex pair remains unchanged. In a similar 
spirit, \citet{hind1996menger}, and independently later also 
\citet{cheriyan2007packing}, introduced the global \emph{reduction lemma}, which 
entails contracting or deleting an edge between Steiner vertices such that the 
element-connectivity of a terminal set remains unchanged. 
\citet{chekuri2014graph} then generalized this to preserve local 
element-connectivity, and this is the starting point for our FPT algorithm for 
\ECSNDP and \LCSNDP.

Finally, we note that several works consider variations of \SNDP where $\dem_{s,t} > 1$ for all $s,t \in X$ for a constant size set $X$, while the demands between all other pairs of vertices in the graph are in $\{0,1\}$ or equal to $1$. For example, \citet{BalakrishnanMM98} gave a $2$-approximation algorithm when $|X|=2$ and all remaining demands are in $\{0,1\}$, assuming the costs satisfy the triangle inequality. \citet{ArkinH08} showed that when $|X|=2$ and all other demands are~$1$, the problem can be solved in polynomial time. They also described approximation algorithms for various other cases that require a specific size of $X$, the demands on pairs in $X$ to be a specified constant, and all other demands to be~$1$.

\section{Tractability of general \SNDP parameterized by solution size}

In this section we prove \cref{thm:LC-alg,thm:VC-alg}. We first show how to 
obtain an FPT algorithm for \LCSNDP, and how this also leads to an algorithm for 
\ECSNDP. For \VCSNDP with $\maxdem\leq 3$ we need to generalize some of the used 
arguments.

\subsection{Element- and edge-connectivity SNDP.}\label{sec:LC-alg}

\citet{chekuri2014graph} proved the so-called \emph{reduction lemma} that 
preserves element-connectivity under deletion or contraction of edges. Here, 
deleting an edge simply means to remove it from the graph, while contracting an 
edge means to identify its incident vertices and removing all resulting loops 
and parallel edges. In the following, for an edge $e\in E$, as usual $G-e$ and 
$G/e$ denote the graph obtained from~$G$ by deleting $e$ and contracting $e$, 
respectively.

\begin{lemma}[LC reduction lemma~\cite{chekuri2014graph}]\label{lem:reduction}
Let $G=(V,E)$ be an undirected graph and $R\subseteq V$ be a terminal set. Let 
$e\in E$ be any edge where $e\cap R=\emptyset$, and let $G_1=G-e$ and 
$G_2=G/e$. Then at least one of the following holds:
\begin{enumerate}[(i)]
\item $\forall s,t\in R:\ \lc_{G_1}(s,t)=\lc_G(s,t)$, or
\item $\forall s,t\in R:\ \lc_{G_2}(s,t)=\lc_G(s,t)$.
\end{enumerate}
\end{lemma}

\citet{chekuri2014graph} remark that their reduction lemma can be applied 
repeatedly until the Steiner vertices form an independent set. By subdividing 
edges between terminals, we may also assume that the terminals form an 
independent set. We will exploit this to prove the following structure of 
\emph{minimal} solutions to \LCSNDP, which are solutions for which no edge can 
be removed without making the solution infeasible. In particular, any optimum 
solution is minimal.

\begin{lemma}\label{lem:tree-partition}
Let $H$ be a minimal solution to an \LCSNDP instance. Then there exist trees 
$T_1,\ldots,T_b\subseteq H$ such that $H=\bigcup_{i=1}^b T_i$, no two trees 
share a Steiner vertex, all leaves of any tree are terminals, and all internal 
vertices of any tree are Steiner vertices. Moreover, for any terminal pair 
$s,t\in R$ there exist $\dem_{s,t}$ element-disjoint paths between $s$ and $t$ 
in $H$, such that the edge set of any tree $T_i$, $i\in\{1,\ldots,b\}$, 
intersects with the edge set of at most one of these paths.
\end{lemma}
\begin{proof}%
To prove the claim, we repeatedly apply \cref{lem:reduction} on the minimal 
solution~$H$. In particular, let $H_0,H_1,\ldots,H_h$ be any sequence of graphs 
we obtain from $H_0=H$ as follows. For each $j\geq 0$ we pick some edge~$e$ of 
$H_j$ where $e\cap R=\emptyset$ and either delete or contract~$e$ without 
decreasing the element-connectivity between any terminal pair, to 
obtain~$H_{j+1}$. Note that $R\subseteq V(H_j)$ for all $j$. According to 
\cref{lem:reduction} we can do this for every~$H_j$, until every edge is 
incident to some terminal.

We claim that in each step $j$ we can only contract the chosen edge $e$ when 
applying \cref{lem:reduction}, as deleting~$e$ will be  decreasing the 
element-connectivity between some terminal pair. The proof is by induction. In 
the first step~$j=0$ this is true since $H_0=H$ is a minimal solution. For a 
step~$j\geq 1$, assume that some edge $e$ of $H_j$ can be deleted without 
reducing the element-connectivity between any terminal pair. This means that 
there exist~$\lc_{H_j}(s,t)$ element-disjoint paths in $H_j$ between each 
terminal pair $s,t\in R$ such that none of these paths contain $e$. By the 
induction hypothesis, some edge $e'$ was contracted in $H_{j-1}$ to obtain 
$H_j$. The element-disjoint paths in $H_j$ also exist in~$H_{j-1}$, after 
uncontracting the edge~$e'$ along some of these paths. Also the edge $e$ exists 
in~$H_{j-1}$ (or an edge corresponding to $e$ after uncontracting~$e'$) but is 
not used by any of the resulting paths in~$H_{j-1}$. But then $e$ can be deleted 
in $H_{j-1}$ without reducing the element-connectivity between any terminal 
pair, contradicting the induction hypothesis.

Now consider the graph $H'$ obtained from $H$ after exhaustively contracting 
edges between Steiner vertices. From above we know that by \cref{lem:reduction}, 
$\lc_{H}(s,t)=\lc_{H'}(s,t)$ for all terminal pairs~$s,t\in R$. Note that in 
$H'$ the remaining Steiner vertices form an independent set. Furthermore, 
w.l.o.g.\ no two terminals are adjacent in $H$ (otherwise we can subdivide such 
an edge using a Steiner vertex). Hence, we may assume that $H'$ is a bipartite 
graph, with $R$ and $V\setminus R$ forming the bipartition. Thus we can 
decompose $H'$ into edge-disjoint stars $S_1,\ldots,S_b$ (which partition the 
edge set of $H'$) such that no two stars share a Steiner vertex.

Now fix $\lc_{H}(s,t)$ element-disjoint paths between each terminal pair $s,t$ 
in $H'$. No edge set of a star $S_i$, $i\in\{1,\ldots,b\}$, can intersect with 
more than one of the $\lc_{H}(s,t)$ paths for any terminal pair~$s,t\in R$, as 
these paths would otherwise intersect in the center vertex of the star, which is 
a Steiner vertex. Note that for each of these $\lc_{H}(s,t)$ paths there exist 
$\lc_{H}(s,t)$ corresponding element-disjoint paths in $H$, by uncontracting 
edges along the paths. Let $T_i$ be the subgraph of~$H$ that is obtained from 
$S_i$ after uncontracting all edges. Also the edge set of~$T_i$ intersects with 
at most one of the $\lc_{H}(s,t)$ paths between any $s$ and $t$. This means that 
if $T_i$ contains a cycle, then we may delete any edge of the cycle, while 
maintaining a path between $s$ and $t$ through~$T_i$ (if the path passed 
through the deleted edge, we may reroute it through the remaining edges of the 
cycle). However this would contradict the minimality of $H$, and thus $T_i$ is a 
tree, which concludes the proof.
\end{proof}

We also need a similar structural result for minimal \ECSNDP solutions. Note 
that in contrast to \cref{lem:tree-partition}, the trees of the following lemma 
are edge-disjoint instead of internally vertex-disjoint.

\begin{lemma}\label{lem:tree-partition2}
Let $H$ be a minimal solution to an \ECSNDP instance. Then there exist trees 
$T_1,\ldots,T_b\subseteq H$ such that $H=\bigcup_{i=1}^b T_i$, no two trees 
share an edge, all leaves of any tree are terminals, and all internal vertices 
of any tree are Steiner vertices. Moreover, for any terminal pair $s,t\in R$ 
there exist $\dem_{s,t}$ edge-disjoint paths between $s$ and~$t$ in $H$, such 
that the edge set of any tree $T_i$, $i\in\{1,\ldots,b\}$, intersects with the 
edge set of at most one of these paths.
\end{lemma}
\begin{proof}
We use a standard reduction\footnote{In the conference version of this 
paper~\cite{DBLP:conf/sosa/Feldmann0L22} it was claimed that the parameter 
$\ell$ only grows by a factor of $2$ in this reduction, which if true would mean 
that it could be used in combination with the FPT algorithm for \LCSNDP to 
obtain an algorithm for \ECSNDP with the same asymptotic running time. However, 
the parameter growth might in fact be quadratic, leading to a worse runtime than 
claimed in \cref{thm:LC-alg} for \ECSNDP using this direct 
approach.\label{footnote:Sasha}} 
from \ECSNDP to \LCSNDP and then invoke \cref{lem:tree-partition}. The reduction 
takes every Steiner vertex $v$ of the input graph $G$ of the instance $I$ to 
\ECSNDP and replaces it by a clique~$K_v$ of size equal to the degree of $v$ in 
$G$, in order to obtain a new graph~$G'$. Two such cliques $K_u$ and $K_v$ are 
then connected by an edge in $G'$ if $uv$ was an edge in $G$, and every terminal 
in $G'$ is connected to the cliques corresponding to its Steiner neighbours 
in~$G$ (edges between terminals are untouched). Note that it is possible to 
obtain these connections in such a way that every vertex of a clique $K_v$ has 
exactly one neighbour outside the clique in $G'$, since the size of the clique 
is equal to the degree of $v$ in $G$. The new instance~$I'$ of \LCSNDP is given 
by $G'$ where the edges of cliques all have cost~0 and all other edges have cost 
corresponding to their edge in the instance $I$. It is now easy to see that 
there is an \LCSNDP solution $H'\subseteq G'$ in the new graph if and only if 
there is an \ECSNDP solution $H\subseteq G$ of the same cost in the original 
graph: to convert a solution $H'$ to a solution in $G$ we simply contract all 
edges that belong to a clique $K_v$, which means that any two element-disjoint 
paths of~$H'$ will not share any edge in the resulting solution $H$ in $G$. To 
convert a solution $H$ to a solution in~$G'$ we can just add all edges of every 
clique~$K_v$, which means that any two edge-disjoint paths of $H$ that meet in a 
Steiner vertex $v$ can be extended by using two vertex-disjoint edges of $K_v$ 
to make the paths element-disjoint.

By \cref{lem:tree-partition} we know that any minimal solution $H'$ to 
$\LCSNDP$ in $G'$ can be decomposed into internally vertex-disjoint trees 
$T'_1,\ldots,T'_b\subseteq H'$. When converting the solution $H'$ to a 
solution $H$ in $G$ as described above, these trees are converted into 
edge-disjoint trees $T_1,\ldots,T_b\subseteq H$ with the required properties.
\end{proof}

Using colour coding and known FPT algorithms for \pname{Steiner Tree}, we 
exploit \cref{lem:tree-partition,lem:tree-partition2} to obtain the following 
result.

\LCalg*
\begin{proof}
We first consider the \LCSNDP problem. Recall that by \cref{lem:tree-partition}, 
an optimum solution $H$ can be partitioned into $b$ internally vertex-disjoint 
trees $T_1,\ldots,T_b$ for some $b$. Note that $b$ is bounded by $\solsize$, and 
recall that the trees only overlap on the terminals. For a terminal $t \in 
\terms$, let $c_t$ denote the number of trees incident on~$t$. Then we observe 
that $f := \sum_{t \in \terms} c_t \leq 2\solsize$, because any terminal is 
incident on a tree $T_i$ via a unique edge of $H$, unless this edge goes between 
two terminals.

We now describe the algorithm. Given an instance to \LCSNDP, we first guess the 
number $b$ of trees into which the optimum solution $H$ can be partitioned 
according to \cref{lem:tree-partition}. Then we guess the total number $f$ of 
edges of the trees incident on the terminals of $G$. We now use colour 
coding~\cite{alon1995color}. First, we randomly colour the Steiner vertices of 
the input graph $G=(V,E)$ with $b$ colours. Then we randomly colour the 
terminals of~$G$ by creating a bin with $\nrterm$ balls of each of the $b$ 
colours (so $\nrterm b$ balls in total), randomly taking $f$ balls from the bin, 
and if the $j$-th ball of colour $i \in [b]$ was taken, assigning the $j$-th 
terminal colour $i$. Note that terminals are thus assigned a set of colours. 
That is, we pick a function~$\varphi$ that maps $V\setminus\terms$ to~$[b]$ and 
$\terms$ to~$2^{[b]}$, such that~$\sum_{t \in \terms} |\varphi(t)| = f$. We then 
condition on the event that every Steiner vertex $v\in V(H)\setminus R$ 
of~$H\subseteq G$ has colour $\varphi(v)=i$ if it belongs to tree $T_i$ of the 
decomposition, and that every terminal~$t\in\terms$ is coloured with a subset 
$\varphi(t)=C$ of $[b]$ such that $i\in C$ if and only if the terminal belongs 
to tree~$T_i$. For every colour $i\in [b]$ we then define the graph $G_i$ 
induced by all vertices of colour $i$ in~$G$, and compute an optimum Steiner 
tree for the terminal set $\terms_i=\{t\in\terms\mid i\in\varphi(t)\}$ of $G_i$. 
Such a tree can be computed using an FPT algorithm for \pname{Steiner Tree} 
parameterized by the number of terminals (e.g., using the 
\citet{dreyfus1971steiner} algorithm). The union of these Steiner trees is the 
computed \LCSNDP solution.

We now argue that the above algorithm outputs an optimum \LCSNDP solution, 
conditioned on the correct colouring. By \cref{lem:tree-partition} there exist 
$\dem_{s,t}$ element-disjoint paths between any terminal pair $s,t\in\terms$ in 
any optimum solution~$H$, such that each tree~$T_i$ of the decomposition 
contains at most one such path. If $x,y\in R$ are terminals of $T_i$ used by 
such a path $P$ between $s$ and $t$ in $H$, then the Steiner tree computed for 
$G_i$ also contains a path between $x$ and $y$. Hence, we can find a path 
corresponding to $P$ in the computed solution that contains edges of a Steiner 
tree computed for $G_i$ if and only if $P$ contains edges of~$T_i$. The colour 
coding ensures that the computed Steiner trees do not share Steiner vertices. 
This means that we can find $\dem_{s,t}$ element-disjoint paths between $s$ and 
$t$ in the union of the computed Steiner trees. Furthermore, the computed 
solution must have minimum cost, since otherwise some tree~$T_i\subseteq H$ 
could be replaced by an optimum Steiner tree of~$G_i$, yielding a feasible 
solution of smaller cost.

The success probability of the above algorithm is the probability with which the 
algorithm picks a colouring~$\varphi$ such that every Steiner vertex~$v$ of $H$ 
has colour $\varphi(v)=i$ if it belongs to tree~$T_i$, and every terminal 
$t\in\terms$ has set of colours~$\varphi(t)=C$ such that $t$ belongs to tree 
$T_i$ if and only if $i\in C$. Note that we may permute the $b$ colours 
arbitrarily without affecting the success probability. Hence, there are $b!\cdot 
b^{n-|V(H)|}$ correct ways to colour the input graph $G$, since there 
are~$n-|V(H)|$ vertices in~$G$ that are not part of $H$, and all of these are 
Steiner vertices that each receive one colour from~$[b]$, while the vertices in 
$H$ are assigned a unique (set of) colours (determined by the permutation). The 
total number of ways to colour the input graph is ${\nrterm b \choose f} 
b^{n-\nrterm}$, since there are ${\nrterm b \choose f}$ possible ways to assign 
sets of colours to the terminals, and $b$ ways to colour each of the $n-\nrterm$ 
Steiner vertices. Thus the success probability is at least 
\[
\frac{b!\cdot b^{n-|V(H)|}}{(\nrterm b)^{f} b^{n-\nrterm}} \geq 
\frac{b!}{\nrterm^{2\solsize}} b^{\nrterm-4\solsize} \geq 2^{-O(\solsize 
\log \solsize)},
\] 
since $b,\nrterm\in\{1,\ldots,\solsize\}$, and $|V(H)|,f\leq 2\solsize$. Using a 
standard argument, we may run this algorithm~$2^{O(\solsize \log \solsize)}\cdot 
n$ times in order to compute the optimum solution with high probability. Since 
the \citet{dreyfus1971steiner} algorithm has a single-exponential runtime in the 
number of terminals, which is upper bounded by the solution size, each optimum 
Steiner tree can be computed in $2^{O(\ell)}\polyn$ time. Furthermore, the 
number of trees to be computed is also at most the solution size, and thus this 
randomized algorithm takes $2^{O(\solsize \log \solsize)}\polyn$ time. 
Alternatively, we may derandomize~\cite{alon1995color} the colour coding 
algorithm to compute the optimum deterministically in $2^{O(\solsize \log 
\solsize)}\polyn$ time.

For \ECSNDP we use essentially the same technique, but instead of colouring 
Steiner vertices we colour edges. That is, based on the existence of $b$ 
edge-disjoint trees $T_1,\ldots,T_b$ as given by \cref{lem:tree-partition2}, we 
pick a function~$\varphi$ that maps $E$ to~$[b]$ and $\terms$ to~$2^{[b]}$, such 
that~$\sum_{t \in \terms} |\varphi(t)| = f$ for the corresponding value $f$. We 
then compute an optimum Steiner tree in each subgraph $G_i$ spanned by edges of 
colour $i$ for terminal set $\terms_i=\{t\in\terms\mid i\in\varphi(t)\}$. An 
analogous argument to above then shows that due to the properties of 
\cref{lem:tree-partition2} the union of these Steiner trees is an optimum 
\ECSNDP solution, conditioned on the event that $\varphi(e)=i$ for each edge $e$ 
of tree $T_i$ and that every terminal $t\in\terms$ is coloured by a subset 
$\varphi(t)=C$ of $[b]$ such that $i\in C$ if and only if~$t$ belongs to $T_i$. 
To calculate the success probability, note that, if $m$ is the number of edges 
of the input graph, there are now $b!\cdot b^{m-\solsize}$ correct ways to 
colour the graph, while the total number of ways to colour the input graph is 
${kb \choose f}b^m$. Thus the success probability is again at least 
$2^{-O(\solsize\log\solsize)}$, leading to the same asymptotic running time as 
above for the FPT algorithm.
\end{proof}

\subsection{Vertex-connectivity SNDP.}\label{sec:VC-alg}

In order to prove \cref{thm:VC-alg} for \VCSNDP with $\maxdem\leq 3$, we extend 
the reduction lemma of \citet{chekuri2014graph} by showing that when contracting 
an edge of~$G$, the vertex connectivity $\vc(s,t)$ between any terminals $s,t$ 
never drops below~$3$ if its connectivity was larger in $G$.

\begin{lemma}[VC reduction lemma]\label{lem:VCreduction}
Let $G=(V,E)$ be an undirected graph and $R\subseteq V$ be a terminal set. 
Let~$e\in E$ be any edge where $e\cap R=\emptyset$, and let $G_1=G-e$ and 
$G_2=G/e$. Then at least one of the following holds:
\begin{enumerate}[(i)]
\item $\forall s,t\in R:\ \vc_{G_1}(s,t)=\vc_G(s,t)$, or
\item $\forall s,t\in R:$
\begin{itemize}
 \item if $\vc_G(s,t)\geq 4$, then 
$\vc_{G_2}(s,t)\in\{\vc_G(s,t)-1,\vc_G(s,t)\}$, and
\item if $\vc_G(s,t)\leq 3$, then $\vc_{G_2}(s,t)=\vc_G(s,t)$.
\end{itemize}
\end{enumerate}
\end{lemma}

\begin{proof}
We extend the proof of \cref{lem:reduction} by \citet{chekuri2014graph}. In 
particular, we use the same setup, which we summarize below. Here a 
\emph{vertex-tri-partition} $(A,B,C)$ of a graph is a partition of its vertex 
set into non-empty parts $A$, $B$, and $C$, such that $B$ \emph{separates} $A$ 
and~$C$, i.e., every path from $A$ to $C$ contains a vertex of $B$. Also, for 
vertices $u$ and $v$, a \emph{$(u,v)$\hy{}separator} is a vertex set~$B$ such 
that there exists a vertex-tri-partition $(A,B,C)$ with $u\in A$ and $v\in C$ 
(in particular, $u$ and $v$ are not contained in the separator). We will 
repeatedly use the well-known fact that the maximum number of 
internally vertex-disjoint paths between $u$ and $v$ is equal to the size of a 
minimum $(u,v)$-separator, i.e., Menger's theorem~\cite{diestel10}.

Note that in \cref{summary} below, the vertex-tri-partitions and separators are 
of the graph $G_1=G-e$. We remark that all of the following properties are true 
regardless of whether we consider vertex- or element-connectivity. In contrast 
to element-connectivity though, as studied by \citet{chekuri2014graph}, for 
vertex-connectivity it may happen that a minimum $(s,t)$-separator for a 
terminal pair $s,t\in R$ contains other terminals (see~\cref{fig:separators}).

\begin{figure}
\centering
\includegraphics[width=0.5\textwidth]{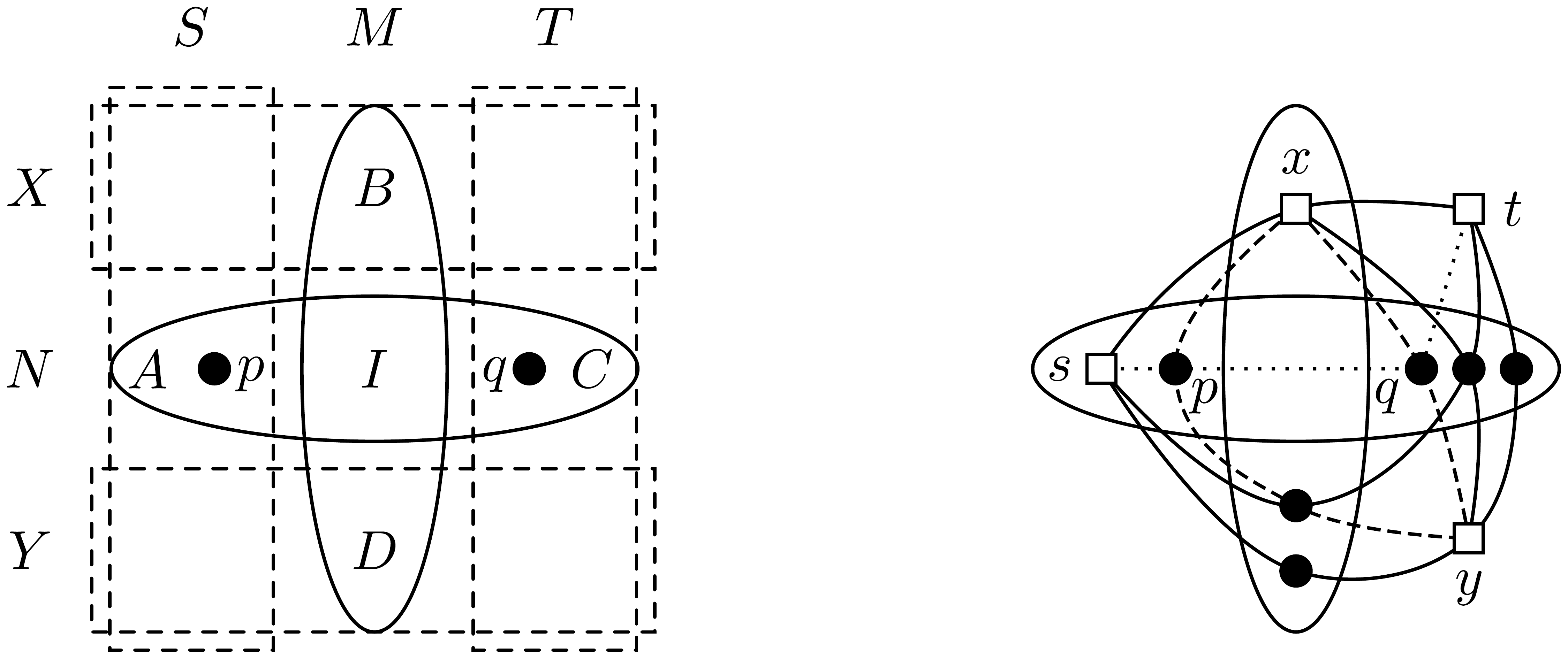}
\caption{The structure of the vertex-tri-partitions of \cref{summary} on the 
left. A possible setup of the terminals (white boxes) on the right, where the 
dotted path does not exist in $G_1$ but in $G$. The dashed paths between $x$ and 
$y$ imply $|B|,|D|\geq 1$, while $s,p\in A$ implies $|A|\geq 2$, and $t\in T\cap 
X$ implies $|C|\geq|D|+1$. In particular, as $|N|\geq |A|+|C|$, in this example 
we have $\vc_G(x,y)\geq 5$.}
\label{fig:separators}
\end{figure}

\begin{prop}[\cite{chekuri2014graph}]\label{summary}
Let $e=\{p,q\}$ be the edge that is deleted or contracted to obtain $G_1$ 
and~$G_2$. Assuming there are terminal pairs $s,t\in R$ and $x,y\in R$ such that 
$\vc_{G_1}(s,t)=\vc_G(s,t)-1$ and $\vc_{G_2}(x,y)=\vc_G(x,y)-1$, the following 
holds:
\begin{enumerate}
\item every minimum $(s,t)$-separator of $G_1$ is also a (not necessarily 
minimum) $(p,q)$-separator of~$G_1$,
\item there is a vertex-tri-partition $(S,M,T)$ of $G_1$ such that $s\in S$, 
$t\in T$, and $M$ is a minimum $(s,t)$-separator in $G_1$ with 
$|M|=\vc_{G_1}(s,t)=\vc_G(s,t)-1$ where, w.l.o.g., $p\in S$ and $q\in T$,
\item there is a vertex-tri-partition $(X,N,Y)$ of $G_1$ such that $x\in X$, 
$y\in Y$, and $N$ is a minimum $(x,y)$-separator in $G_1$ with 
$|N|=\vc_{G_1}(x,y)=\vc_G(x,y)$ where $p,q\in N$, and
\item if $A,B,C,D$, and $I$ respectively denote $S\cap N$, $X\cap M$, $T\cap N$, 
$Y\cap M$, and $N\cap M$, we have $p\in A$ and $q\in C$, and
\begin{itemize}
\item if $s\in S\cap X$ then $|A|\geq |D|+1$,
\item if $s\in S\cap Y$ then $|A|\geq |B|+1$,
\item if $t\in T\cap X$ then $|C|\geq |D|+1$,
\item if $t\in T\cap Y$ then $|C|\geq |B|+1$.
\end{itemize}
\end{enumerate}
\end{prop}

\begin{proof}
For completeness we provide a short proof of this proposition as can also be 
found in~\cite{chekuri2014graph}. We proceed point by point:
\begin{enumerate}
 \item holds since otherwise it would be an $(s,t)$-separator of size 
$\vc_{G_1}(s,t)$ in $G$ so that $\vc_G(s,t)\leq\vc_{G_1}(s,t)$,
\item holds since $M$ is $(p,q)$-separator due to point~1,
\item holds since in $G_2$ every minimum $(x,y)$-separator $N'$ of size 
$\vc_{G_2}(x,y)=\vc_G(x,y)-1$ must contain the vertex into which $e$ is 
contracted, and so any $(x,y)$-separator of $G_1$ has size at 
least~$\vc_G(x,y)$; the one in $G_1$ corresponding to $N'$ contains $p$ and $q$, 
has size $\vc_G(x,y)$, and is thus minimum,
\item $p\in A$ and $q\in C$ holds due to the previous points, and
\begin{itemize}
 \item holds since then $A\cup I\cup B$ is an $(s,t)$-separator, but not a 
minimum one due to point~1 as $p\in A$, in contrast to $M=B\cup I\cup D$, so 
that $|B\cup I\cup D|\leq|A\cup I\cup B|-1$,
\item ditto, since then $A\cup I\cup D$ is a non-minimal $(s,t)$\hy{}separator 
as $p\in A$,
\item ditto, since then $B\cup I\cup C$ is a non-minimal $(s,t)$\hy{}separator 
as $q\in C$,
\item ditto, since then $C\cup I\cup D$ is a non-minimal $(s,t)$\hy{}separator 
as $q\in C$.%
\end{itemize}
\end{enumerate}
\end{proof}

Now assume that the statement of the lemma does not hold. This means that there 
exists a terminal pair~$s,t\in R$ for which $\vc_{G_1}(s,t)\neq\vc_G(s,t)$ and 
a terminal pair $x,y\in R$ for which either
\begin{itemize}
\item $\vc_G(x,y)\geq 4$ and 
$\vc_{G_2}(x,y)\notin\{\vc_G(x,y)-1,\vc_G(x,y)\}$, or
\item $\vc_G(x,y)\leq 3$ and $\vc_{G_2}(x,y)\neq\vc_G(x,y)$.
\end{itemize}
Observe that deleting or contracting~$e$ can only decrease the connectivity 
between any terminal pair by $1$, and thus $\vc_{G_1}(s,t)=\vc_G(s,t)-1$ 
and $\vc_{G_2}(x,y)\in\{\vc_G(x,y)-1,\vc_G(x,y)\}$. Hence, only the 
latter of the two cases can apply to $x$ and $y$, i.e., we have 
$\vc_G(x,y)\leq 3$ and $\vc_{G_2}(x,y)=\vc_G(x,y)-1$. Furthermore, we 
obtain the properties of \cref{summary}. We will show that both $B$ and $D$ 
contain at least one vertex each, and as a consequence $A$ and $C$ contain at 
least two vertices each. The latter implies that $|N|\geq |A|+|C|\geq 4$ and 
thus $\vc_G(x,y)=\vc_{G_1}(x,y)\geq 4$ --- a contradiction.

To prove that $B,D\neq\emptyset$, fix $\vc_{G_1}(x,y)$ internally 
vertex-disjoint paths between $x$ and $y$ in~$G_1$. By point~3 of 
\cref{summary}, there are two such paths~$P,Q$ such that $P$ contains $p\in N$ 
and $Q$ contains~$q\in N$ (cf.~\cref{fig:separators}). Since each of these 
paths can only contain one vertex of the minimum $(x,y)$\hy{}separator~$N$ (due 
to Menger's theorem), the union of these two paths forms a cycle containing two 
paths from $p$ to $q$, such that the internal vertices of one path are only from 
$X$ (including~$x$) and the internal vertices of the other are only from $Y$ 
(including~$y$). Since $M$ separates $p$ and~$q$ by point~1 of \cref{summary}, 
the first path contains a vertex of $B=M\cap X$, while the second one contains a 
vertex of $D=M\cap Y$, i.e., both $B$ and $D$ are non-empty.

This implies that $A$ and $C$ contain two vertices each, as follows. If $t\in C$ 
then~$|C|\geq 2$, since $q\in C$ by point~4 of \cref{summary}, but $q\neq t$ as 
$q$ is not a terminal in contrast to $t$. Otherwise, $t\in T\cap X$ or $t\in 
T\cap Y$, which by \cref{summary} means that $|C|\geq |D|+1$ or $|C|\geq |B|+1$. 
From above we know that $|B|\geq 1$ and $|D|\geq 1$ and so~$|C|\geq 2$ in either 
case. Analogously for~$A$, if $s\in A$ we are done as $p\in A$ and $p\neq s$, 
and if $s\notin A$ we get $|A|\geq|B|+1$ or~$|A|\geq|D|+1$ by \cref{summary}. 
Hence, we also have $|A|\geq 2$, which concludes the proof.
\end{proof}

As a consequence, we obtain the same structure of minimal solutions for \VCSNDP 
with $\maxdem\leq 3$ as for \LCSNDP, by simply replacing the application of 
\cref{lem:reduction} by \cref{lem:VCreduction} in the proof of 
\cref{lem:tree-partition} to show the following.

\begin{lemma}\label{lem:FVS}
Let $H$ be a minimal solution to a \VCSNDP instance with maximum demand at 
most~$3$. Then there exist trees $T_1,\ldots,T_b\subseteq H$ in $H$ such that 
$H=\bigcup_{i=1}^b T_i$, no two trees share a Steiner vertex, all leaves of any 
tree are terminals, and all internal vertices of any tree are Steiner vertices. 
Moreover, for any terminal pair $s,t\in R$ there exist $\dem_{s,t}$ 
vertex-disjoint paths between $s$ and~$t$ in $H$, such that the edge set of any 
tree $T_i$, $i\in\{1,\ldots,b\}$, intersects with the edge set of at most one of 
these paths.
\end{lemma}

\cref{lem:FVS} implies the same FPT algorithm parameterized by the solution size 
as for \LCSNDP as found in the proof of \cref{thm:LC-alg}, and thus we obtain 
the following theorem.

\VCalg*

\newcommand{\linc}{\mbox{inc}}
\newcommand{\lpath}{\mbox{path}}
\newcommand{\ldisjoint}{\mbox{disjoint}}

\section{Tractability of \VCSNDP for treewidth and planar graphs}\label{sec:tw}

We first prove that {\VCSNDP} is FPT parameterized by the treewidth $\tw$ plus 
the sum of demands~$\sumdem$. To this end, we provide a simple formulation 
of the problem in MSOL and apply a result by \citet{ArnborgLS91}. We then argue 
that this implies that {\VCSNDP} is FPT on planar graphs parameterized by the 
solution size $\solsize$.

For sake of completeness, we define treewidth and describe the result of 
\citet{ArnborgLS91} that we rely on.

\begin{Definition}
A \emph{tree decomposition} of a graph $G=(V,E)$ is a tree $T$ and a family $B$ of subsets of vertices (called bags), one bag $B(t)$ per vertex $t$ of the tree, such that the following holds:
\begin{enumerate}
\item $\bigcup_{t \in V(T)} B(t) = V$;
\item for each edge $uv \in E$, there is a $t \in V(T)$ such that $u,v \in B(t)$;
\item for each vertex $v \in V$, the vertices $t \in V(T)$ for which $v \in B(t)$ induce a subtree of $T$.
\end{enumerate}
The \emph{width} of a tree decomposition is $\max_{t\in V(T)} |B(t)|-1$. The 
\emph{treewidth} of a graph is the minimum width over any of its tree 
decompositions.
\end{Definition}

The result of Arnborg et al.~\cite[Theorem 5.6]{ArnborgLS91} applies Monadic 
Second Order Logic (MSOL) to graphs. An MSOL formula is a logical formula, 
allowing universal and existential quantifiers over variables that are single 
vertices or edges, variables that are sets of vertices or edges, basic logical 
formulas $(\vee,\wedge,\lnot,\rightarrow)$, basic binary relations $(\in,=)$, 
and the special binary relation $\linc(v,e)$ that is true if and only if edge 
$e$ is incident on vertex~$v$. Note that using the primitives of MSOL, it is 
easy to define other common primitives, such as $\not\in, \not=, \subseteq$, on 
the way to creating more complex formulas.

\begin{theorem}[{\cite[Theorem 5.6]{ArnborgLS91}}] \label{thm:msol}
Let $G=(V,E)$ be an $n$-vertex graph of treewidth $tw$ with a constant number of special sets $S_1,\ldots,S_q$ of vertices or edges. Let $\phi$ be an MSOL formula that uses $S_1,\ldots,S_q$ and with $p$ free set variables $X_1,\ldots,X_p$, let $f_1,\ldots,f_p$ be $p$ functions with the same domain as $X_1,\ldots,X_p$ respectively that assign integer values, and let $F$ be a linear function on $p$ variables. Then there is an algorithm that finds in $f(tw,|\phi|) \cdot n$ time the minimum or maximum value of $F(\sum_{x \in X_1} f_1(x),\ldots,\sum_{x \in X_p} f_p(x))$ for sets $X_1,\ldots,X_p$ such that $\phi(X_1,\ldots,X_p)$ is satisfied, for some function $f$.
\end{theorem}

We now formulate {\VCSNDP} as an MSOL formula to obtain the following result.

\twDalg*
\begin{proof}
In this proof, we only consider terminal pairs $s,t \in \terms$ for which $\dem_{s,t} > 0$ (often implicitly).
Let $z$ be the number of terminal pairs $s,t \in \terms$ for which $\dem_{s,t} > 0$, so $z \leq {\nrterm \choose 2}$. Order the pairs of distinct terminals of $\terms$ arbitrarily. Use $S_1,\ldots,S_z$ to denote the sets of terminal pairs, that is, $S_i = \{s_i,t_i\}$ where $s_i,t_i$ is the $i$-th terminal pair (with $\dem_{s_i,t_i} > 0$).
We define the following primitive to describe that, given a set of vertices $V_P$ and a set of edges $X$, is true if and only if there exists a path between the terminal pair $s_i,t_i$ on the vertices of $V_P$ that uses a subset of $X$.
\[
\lpath_{i}(V_P,X) = &\ S_i \subseteq V_P \wedge \forall Z \subseteq V_P \Big( 
\big(\exists u \in S_i: u \in Z \wedge \exists v \in S_i : v \not\in Z\big) 
\rightarrow \\&
\big(\exists e \in X\ \exists u,v \in V_P : \linc(u,e) \wedge 
\linc(v,e) \wedge u \in Z \wedge v \not\in Z\big)\Big).
\]
Note that this formula ensures that $s,t \in V_P$ and that for any partition of 
$V_P$ that separates $s_i$ and~$t_i$, there is an edge of $X$ crossing the 
partition. Hence, the subgraph $(V_P, X \cap (V_P \times V_P))$ contains an 
$s_i$-$t_i$ path as required (by Menger's theorem~\cite{diestel10}).

We also define a primitive $\ldisjoint_{i}(V_1,\ldots,V_j)$, which is true if 
and only if the input sets $V_1,\ldots,V_j$ are pairwise disjoint except that 
all sets contain $S_i$ as a subset. We omit the straightforward definition.

We can then use these primitives for the following function:
\[
\phi(X) =&\ X \subseteq E \wedge \forall i \in\{1,\ldots,z\}\ \exists 
V_1,\ldots,V_{\dem_{s_i,t_i}} \subseteq V :\\ 
&\left(\ldisjoint_{i}(V_1,\ldots,V_{\dem_{s_i,t_i}}) \wedge \forall j \in 
\{1,\ldots,\dem_{s_i,t_i}\} : \lpath_{i}(V_j,X)\right).
\]
Here we are using $\forall i \in \{1,\ldots,z\} : \phi_i$ as a shortcut for 
$\phi_1 \wedge \cdots \wedge \phi_z$. Each $\phi_i$ must be explicitly written 
out, because the demands $\dem_{s_i,t_i}$ are not an input to the formula or 
the graph structure. This is also the reason why we need the sets $S_i$. 
Therefore, the length of $\phi(X)$ depends only on~$\sumdem=\sum_{i=1}^z 
\dem_{s_i,t_i}$.

Assume now that a tree decomposition is known; otherwise, compute one using 
Bodlaender's algorithm~\cite{Bodlaender96}. Then, apply the minimization 
version of \cref{thm:msol} on $\phi$ with $f_1$ being the edge cost function, 
$F$ the identity function, and $S_1,\ldots,S_z$ being the special sets. The 
theorem follows.
\end{proof}

We now show that our result extends to graphs of bounded local treewidth for 
parameter solution size. A graph has \emph{bounded local treewidth} if there is 
a function $g$ such that the treewidth of any subgraph induced by the vertices 
within (shortest path) distance $r$ of any vertex is at most $g(r)$.
\citet{AlberBFKN02} showed that planar graphs have bounded local treewidth for a 
linear function $g$. \citet{Eppstein00} proved that the graphs of bounded local 
treewidth are exactly the apex-minor-free graphs (graphs that exclude a fixed 
apex graph as a minor) and \citet{DemaineH04} showed that the function $g$ is 
linear here as well.

\ltwalg*
\begin{proof}
Observe that all edges of any solution must be within (shortest path) distance at most 
$\solsize$ of some terminal. Hence, we may restrict the graph to all vertices 
and edges within distance $\solsize$ of any terminal and remove all others. 
The resulting graph has diameter $\solsize\cdot \nrterm \leq \solsize^2$ and 
thus has treewidth bounded as a function of $\solsize$. The result then follows 
from \cref{thm:twD-alg}, as $\sumdem\leq 2\solsize$.
\end{proof}

\section{Hardness of \VCSNDP}\label{sec:ss}

In this section we prove \cref{thm:VC-hard}, i.e., that uniform single-source 
\VCSNDP with $\nrterm=3$ terminals is W[1]-hard parameterized by the solution 
size~$\solsize$. We give a reduction from the \pname{Grid Tiling} problem, where 
the input consists of integers $K, n$ and a collection of $K^2$ non-empty sets 
$S_{i,j} \subseteq [n] \times [n]$ of integer pairs for $i, j \in [K]$. We think 
of such an instance as a $K\times K$ grid, where each grid cell contains a set 
$S_{i,j}$ of integer pairs. The goal is to find one ordered pair $s_{i,j} \in 
S_{i,j}$ for every $i,j\in [K]$, such that if $(i,j)$ and~$(i',j')$ are adjacent 
in the first or the second coordinate, then $s_{i,j}$ and $s_{i',j'}$ agree in 
the first or the second coordinate, respectively. More formally,
\begin{itemize}
\item if $s_{i,j}=(a,b)$ and $s_{i+1,j}=(a',b')$ then $a=a'$, and
\item if $s_{i,j}=(a,b)$ and $s_{i,j+1}=(a',b')$ then $b=b'$.
\end{itemize}
If such pairs exist, the instance is a \emph{YES-instance}, and otherwise it is 
a \emph{NO-instance} (i.e., this is a decision problem).
\pname{Grid Tiling} is known to be \textup{W[1]}-hard when parameterized by $K$ 
\cite{Marx12}. 

In the following we first give a polynomial time construction of a \VCSNDP 
instance given a \pname{Grid Tiling} instance, followed by the proof of 
correctness of our reduction. 

\subsection{Construction.}\label{subsec:construct}
Given an instance for \pname{Grid Tiling} we first create a 
graph $G = (V, E)$ with $V=\bigcup_{i,j\in[K]}V_{i,j}$ and 
$E=\bigcup_{i,j\in[K]}E_{i,j}$ as follows.  For each pair $(a,b) \in S_{i,j}$, 
add two new vertices $v_{(a,b)}, v^*_{(a,b)}$ to $V_{i,j}$ and put 
an edge between them, i.e., $\{v_{(a,b)}, v^*_{(a,b)}\} \in E_{i,j}$. For $i,j 
\in [K]$, we call the graph~$(V_{i,j}, E_{i,j})$ a \emph{cell}~$c_{i,j}$ of 
$G$, and we use the notation $V^*_{i,j} = \bigcup_{(a,b)\in S_{i,j}} 
\{v^*_{(a,b)}\}$. Note that each cell $c_{i,j}$ is simply a perfect matching 
between $V^*_{i,j}$ and~$V_{i,j}\setminus V^*_{i,j}$.
We also put edges across these cells if they are adjacent, i.e., we add an edge 
between $v^*_{(a,b)} \in V^*_{i,j}$ and $v_{(a',b')} \in V_{i',j'}\setminus 
V^*_{i',j'}$ whenever $|i' -i|+|j'-j| =1$, $i' \geq i$, $j' \geq j$, and if $i=i'$ then $a=a'$, or if 
$j=j'$ then~$b=b'$ (cf.~\cref{fig:2}).

\begin{figure}
\centering
\includegraphics[scale=0.5]{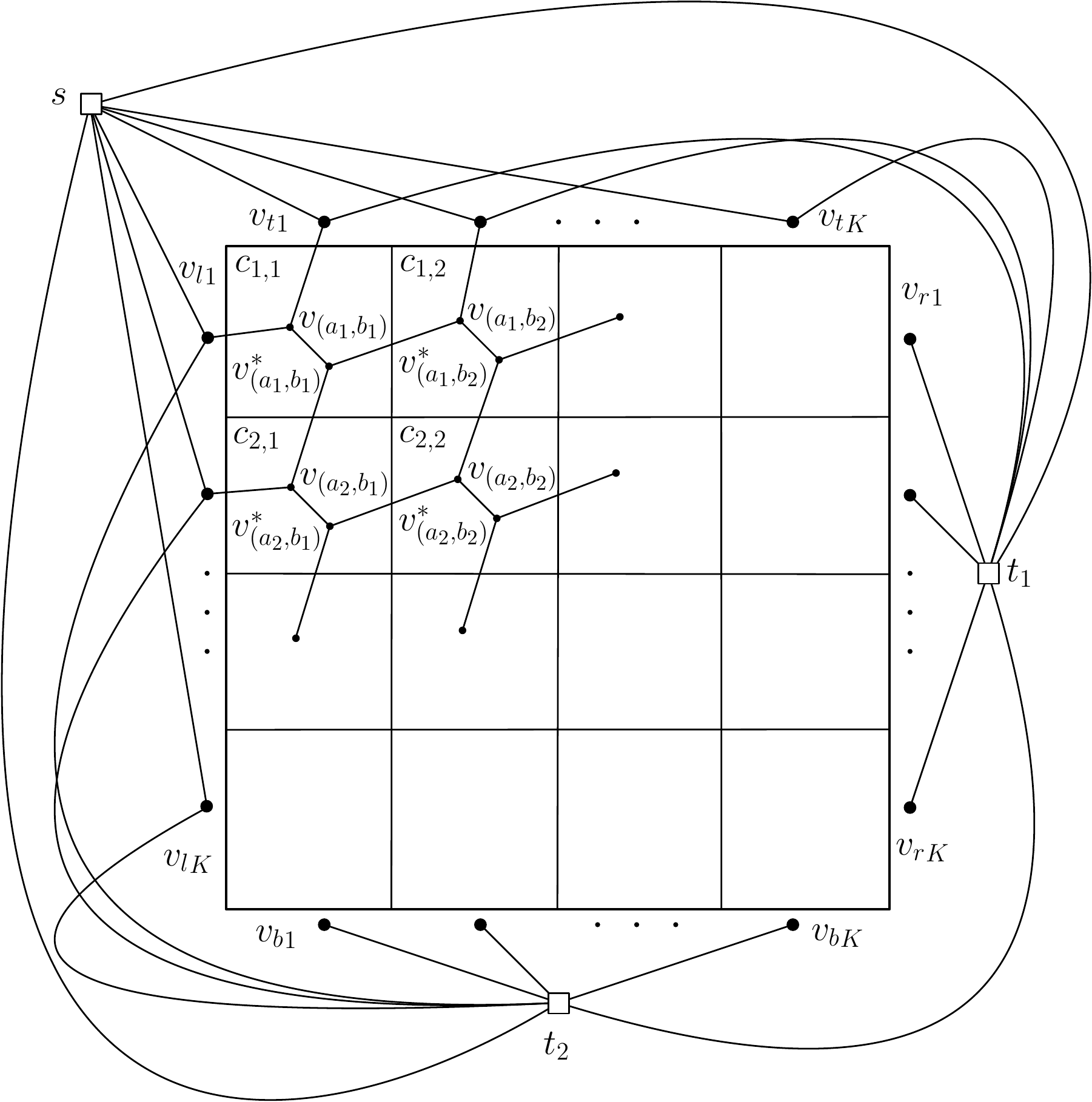}
\caption{Graph $G' = (V', E')$ with source $s$ and sinks $t_1, t_2$. The 
annotated cells depict some edge connections inside and across the cells.}
\label{fig:2}
\end{figure}

Next we do the following set of modifications to the graph $G$:
\begin{enumerate}
\item Create $4$ new sets of vertices $V_{(l)}, V_{(r)}, V_{(t)}, V_{(b)}$ each of size $K$ where $v_{xi}$ belongs to $V_{(x)}$ for $x \in \{l,t,r,b\}$ and $i \in [K]$.
\item For $i,j \in [K]$, add all possible edges between $v_{li}$ and~$V_{i,1} 
\setminus V^*_{i,1}$, and all possible edges between~$v_{tj}$ and $V_{1,j} 
\setminus V^*_{1,j}$. Similarly, add all possible edges between $v_{ri}$ and 
$V^*_{i,K}$, and also between $v_{bj}$ and~$V^*_{K,j}$. 
\item Add $3$ new vertices $\terms =\{s,t_1,t_2\}$, which form the terminal set.
\item Connect $t_1$ to all the vertices in $V_{(t)} \cup V_{(r)}$ and connect $t_2$ to $V_{(l)} \cup V_{(b)}$.
\item Connect $s$ to all the vertices in $V_{(t)}$ and $V_{(l)}$.
\item Add an edge each between $s$ and $t_1$, between $s$ and $t_2$, and also 
between $t_1$ and $t_2$.
\end{enumerate}

Let $G' = (V', E')$ be the resulting graph with $V' = V \cup \terms 
\cup V_{(l)} \cup V_{(r)}\cup V_{(t)}\cup V_{(b)}$. The target \VCSNDP instance 
contains $G'$, where the terminals are $\terms = \{s,t_1, t_2\}$ and the 
demands between the source $s$ and the sinks~$t_1, t_2$ are $d_{s,t_1} = 
d_{s,t_2} = 2K+2$. Note that this is a uniform single-source instance with 
root~$s$. All the edges~$e \in E'$ are undirected and of unit cost, i.e., 
$\cost(e) =1$.  This completes the description of our construction. It is easy 
to see that it can be done in polynomial time.

\subsection{Correctness.}

To prove the correctness of our reduction we show that the \VCSNDP instance has 
a solution~$H\subseteq G'$ of cost at most $3K^2 + 8K + 3$ 
if and only if the given \pname{Grid Tiling} instance is a YES-instance. 

Before moving on to establish the correctness, we first define the following 
notions. Call the vertices $V_{(x)} $ for $x \in \{l,t,r,b\}$ as \emph{boundary 
vertices}. Any path connecting a vertex from $V_{(l)}$ with a vertex from 
$V_{(r)}$ is called a \emph{horizontal path}, and a path connecting~$V_{(t)}$ 
with $V_{(b)}$ is a \emph{vertical path}. Call the subgraph induced by vertices 
in~$\bigcup_{j\in [K]} V_{i,j}$ the \emph{horizontal layer} $i$ and the subgraph 
induced by vertices in $\bigcup_{i \in [K]} V_{i,j}$ the \emph{vertical 
layer}~$j$. Note that horizontal paths do not necessarily lie in a horizontal 
layer, and also vertical paths are not bound to any vertical layer. We call the 
edges corresponding to pairs in~$S_{i,j}$ (i.e., edges of all $E_{i,j}$) the 
\emph{cell edges} and edges going across cells between adjacent layers the 
\emph{connector edges}. Note that the horizontal layer $i$ and the vertical 
layer $j$ only intersect in the cell edges of $c_{i,j}$. Any connector edge is 
only contained in either a horizontal or a vertical layer, and we 
correspondingly refer to the former connector edges as \emph{horizontal} and the 
latter as \emph{vertical}.

We are now ready to prove the easy part of the reduction, which is captured in 
the following lemma.

\begin{lemma}\label{lem:ssub}
If the \VCSNDP instance is constructed for a YES-instances of \pname{Grid 
Tiling}, then there is solution $H\subseteq G'$ with $\cost(H)=3K^2 + 8K + 3$.
\end{lemma}
\begin{proof}
Since the given \pname{Grid Tiling} instance is a YES-instance, for each $i,j 
\in [K]$, we have a pair $s_{i,j} \in S_{i,j}$ such that if $(i,j)$ and 
$(i',j')$ are adjacent in the first or the second coordinate, then 
$s_{i,j}=(a,b)$ and $s_{i',j'}=(a',b')$ agree in the first or the second 
coordinate, respectively. In our \VCSNDP instance each pair $s_{i,j}$ 
corresponds to a unique cell edge $\{v_{(a,b)},v^*_{(a,b)}\}$ in the 
cell~$c_{i,j}$. Furthermore, by construction, in each horizontal (resp., 
vertical) layer these cell edges of two adjacent cells are joined via a 
connector edge,  due to the agreement of the corresponding pairs on the first 
(resp., second) coordinate. 

We construct a solution $H$ for the \VCSNDP instance as follows. For each $i 
\in [K]$,  we have a horizontal path connecting $v_{li}$ to $v_{ri}$ strictly 
through the horizontal layer $i$, alternately using cell edges and horizontal 
connector edges.  Similarly, for each $j \in [K]$,  we have a path connecting 
$v_{tj}$ to~$v_{bj}$ through the vertical layer $j$, alternately using cell 
edges and vertical connector edges. Note that we can find such paths so that 
a single cell edge is used in each cell $c_{i,j}$, which is shared by the $i$-th 
horizontal and the $j$-th vertical path. Altogether these amount to $K^2$ cell 
edges, $2K(K-1)$ connector edges joining them inside the grid, and an 
additional $4K$ edges connecting the boundary vertices to the outermost cell 
edges of the grid. To complete the solution we further add the~$3$ direct edges 
between the terminals, and also all the $6K$ edges between the terminals and 
the boundary vertices. As all edges have unit cost, the cost of this solution is 
exactly $K^2 + 2K(K-1) + 4K + 3+ 6K= 3K^2 + 8K +3$.

For correctness, we have $K$ horizontal paths, one in each horizontal layer $i$
(and hence, pairwise vertex-disjoint), which must exist due to the agreement of 
some $K$ pairs $s_{i,j}$ (one corresponding to each cell $c_{i,j}$) on their 
first coordinate in each row in the \pname{Grid Tiling} instance.  Similarly, 
there are $K$ vertex-disjoint vertical paths, one in each vertical layer $j$, 
due to the agreement on the second coordinate in each column in the 
\pname{Grid Tiling} instance. This means that there are $K$ vertex-disjoint 
paths from $s$ to $t_1$ via $V_{(l)}$ and $V_{(r)}$ in $H$, and also $K$ 
vertex-disjoint paths from $s$ to $t_2$ via $V_{(t)}$ and~$V_{(b)}$. In 
addition, there are $K$ vertex-disjoint paths from $s$ to~$t_1$, each of length 
two, via $V_{(t)}$ that are vertex-disjoint from the former paths, and also $K$ 
vertex-disjoint paths from $s$ to~$t_2$, each of length two, via $V_{(l)}$ 
disjoint from the former. Hence, there are a total of $d_{s,t_1} = 2K+2$ 
vertex-disjoint paths between $s$ and $t_1$, after adding the path of length two 
via $t_2$ and the path of length one from $s$ to $t_1$. Similarly there are also 
$d_{s,t_2}= 2K+2$ vertex-disjoint paths between $s$ and~$t_2$.
\end{proof}

In the rest of this section, we prove the other direction of the reduction, as 
follows. First in \cref{lem:sslb} we identify some basic structure of any 
solution of cost $3K^2 + 8K +3$ to the \VCSNDP instance. Next, we argue that the 
proposed solution in \cref{lem:ssub} is the only way to achieve the demands with 
this cost. In particular, in \cref{lem:ss2} we prove that any other way of 
connecting the terminals to meet the demands must incur a total cost of more 
than $3K^2 + 8K +3$. 
Note that the solution size is~$\solsize =\cost(H)$ for any solution $H$, since 
all edges have unit cost. 
Hence, the parameter~$\solsize$ is bounded by $O(K^2)$ and 
after establishing the above mentioned lemmas we obtain the following.

\VChard*

We start with two easy claims about the structural properties of the instance 
and the solution. In the following, we consider the sets $V_{(l)}, V_{(t)}$ as 
vertical and horizontal layer $0$, respectively, and the sets~$V_{(r)}, V_{(b)}$ 
as vertical and horizontal layer $K+1$, respectively. This also means that the 
edges between layers $0$ and $1$, and between layers $K$ and $K+1$, are 
considered connector edges.

\begin{claim}\label{claim:dis}
Inside any cell $c_{i,j}$ where $i,j \in [K]$, vertices adjacent to horizontal 
layer $i-1$ and vertical layer $j-1$ are disjoint from the set of vertices 
adjacent to horizontal layer $i+1$ and vertical layer $j+1$. The former set of 
vertices is exactly $V_{i,j}\setminus V^*_{i,j}$, while the latter is exactly 
$V^*_{i,j}$, and they are only connected via the cell edges of $c_{i,j}$.
\end{claim}
\begin{proof}
Follows directly from the construction of the \VCSNDP instance in \cref{subsec:construct}.
\end{proof}

\begin{claim}\label{claim:out}
Any solution to the \VCSNDP instance contains all edges incident to the three 
terminals $s,t_1,t_2$. Moreover, in every solution there are $K$ vertex-disjoint 
horizontal paths $P_1,\ldots,P_K$ that do not go via $t_2$ or~$V_{(t)}$, and $K$ 
vertex-disjoint vertical paths $Q_1,\ldots,Q_K$ that do not go via $t_1$ or 
$V_{(l)}$. 
\end{claim}
\begin{proof}
Note that the degrees of the terminals are all equal to the demand $2K+2$. Hence 
we need to add all edges incident to the terminals to any solution. Considering 
the terminal pair $s,t_1$, the edges incident to the terminals already include 
$K+2$ paths: the $K$ paths of length two via the vertices in $V_{(t)}$, the path 
of length two via $t_2$, and the edge between $s$ and $t_1$. The only remaining 
paths between $s$ and $t_1$ must use the boundary vertices $V_{(l)}$ and 
$V_{(r)}$ and are thus the vertex-disjoint horizontal paths $P_1,\ldots,P_K$. As 
these are vertex-disjoint from the $K+2$ paths implied by edges incident to the 
terminals, they cannot use $t_2$ or $V_{(t)}$. The argument for the 
vertex-disjoint vertical paths $Q_1,\ldots,Q_K$ is analogous.
\end{proof}

In the following lemmas we heavily exploit the grid structure of the graph $G'$ 
in order to control the routes of the horizontal and vertical paths given by 
\cref{claim:out}. This is crucial in our reduction for undirected graphs, as 
alluded to in \cref{sec:techniques}. In the following, let $P_1,\ldots,P_K$ and 
$Q_1,\ldots,Q_K$ denote the paths given by \cref{claim:out} for any solution.

\begin{lemma}\label{lem:sslb}
Any solution to the \VCSNDP instance has cost at least $3K^2 + 8K +3$. If the 
solution has cost exactly $3K^2 + 8K +3$, then (1)~every vertical (horizontal) 
layer $j\in\{1,\ldots,K\}$ contains exactly $K$ edges, each being part of a 
horizontal path $P_1,\ldots,P_K$ (vertical path $Q_1,\ldots,Q_K$), and 
(2)~between adjacent vertical (horizontal) layers $j\in\{0,\ldots,K\}$ and 
$j+1$ there are exactly $K$ edges, each being part of a horizontal path 
$P_1,\ldots,P_K$ (vertical path~$Q_1,\ldots,Q_K$).
\end{lemma}
\begin{proof}
By \cref{claim:out}, no horizontal path $P_1,\ldots,P_K$ contains $t_2$ or any 
vertex from $V_{(t)}$. Consider the graph~$G''$ obtained by removing all 
terminals and $V_{(t)}$ from $G'$, and note that any vertex~$v_{bj}\in V_{(b)}$ 
is only adjacent to~$V^*_{K,j}$~in this graph. Now, due to \cref{claim:dis}, in 
$G''$ each of the two sets $\bigcup_{i=1}^K V^*_{i,j}$ and $\bigcup_{i=1}^K 
V_{i,j}\setminus V^*_{i,j}$ of a vertical layer~$j$ forms a vertex cut between 
$V_{(l)}$ and $V_{(r)}$. Thus the vertex disjoint paths $P_1,\ldots,P_K$ 
connecting $V_{(l)}$ and $V_{(r)}$ in~$G''$ must use at least $K$ vertices from 
each such vertex cut. As a consequence these paths must also use at least $K$ 
edges from each vertical layer and at least $K$ edges between adjacent vertical 
layers. An analogous argument shows that the $K$ vertical paths must use at 
least $K$ edges from each horizontal layer and at least $K$ edges between 
adjacent horizontal layers.

Counting the number of edges used by a solution, between any adjacent vertical 
and horizontal layers (including layers $0$ and $K+1$) there are $K$ edges of 
the solution giving $2(K+1)K$ edges, while each of the $K$ vertical layers also 
contains $K$ edges giving $K^2$ more edges (the horizontal layers intersect the 
vertical ones in cells and thus may not contribute more edges). Additionally, 
by \cref{claim:out} any solution uses all $6K+3$ edges incident to the 
terminals. This amounts to a total of at least $2(K+1)K+K^2+6K+3 = 3K^2 + 8K 
+3$ edges, which is also the cost of the solution as each edge has cost~$1$.
Given the above structural considerations of the horizontal and vertical paths, 
this proves the first part of the lemma, but also proves points (1) and (2), 
since in a solution of exactly this cost we cannot afford any more edges than 
those that we argued are always present.
\end{proof}

Finally, we show that in a solution of cost $3K^2 + 8K +3$, the $K$ horizontal 
(resp., vertical) paths through the grid must join the corresponding pairs from 
$V_{(l)}$ and $V_{(r)}$ (resp., $V_{(t)}$ and $V_{(b)}$) through distinct 
layers: although it might seem possible to route these paths in other ways, 
e.g., a horizontal path connecting $v_{li}$ and $v_{ri'}$ for $i \ne i'$, we 
show that such a solution must use some extra cell edges or connector edges and 
have a total cost of more than~$3K^2 + 8K +3$.

\begin{lemma}\label{lem:ss2}
For a solution of cost $3K^2 + 8K +3$ to the \VCSNDP instance and any $i\in[K]$, 
w.l.o.g., the horizontal path $P_i$ (vertical path $Q_i$) connects $v_{li}$ and 
$v_{ri}$ ($v_{ti}$ and $v_{bi}$) using only vertices of the horizontal layer~$i$ 
(vertical layer~$i$).
\end{lemma}

\begin{proof}%
Since the horizontal paths $P_1,\ldots,P_K$ are vertex disjoint, for any 
$i\in[K]$, w.l.o.g., each path~$P_i$ starts in vertex $v_{li}$ and ends in a 
distinct vertex of $V_{(r)}$ (potentially not $v_{ri}$). In particular, a path 
$P_i$ does not contain any other vertices from $V_{(l)}$ and $V_{(r)}$. By 
\cref{claim:out}, $P_i$ also does not contain any vertex from $V_{(t)}$, nor the 
terminal $t_2$. This means that the internal vertices of $P_i$ are all inside 
the grid or from $V_{(b)}$. Now assume that in some vertical layer $j$ the path 
$P_i$ moves \emph{downwards} from some horizontal layer $i'\in[K]$ to horizontal 
layer $i'+1$. If $i'\leq K-1$ so that horizontal layer $i'+1$ is not the set 
$V_{(b)}$, then due to~\cref{claim:dis}, the path $P_i$ can only arrive at a 
vertex from $V_{i'+1,j}\setminus V^*_{i'+1,j}$ of cell $c_{i'+1,j}$ from a 
vertex of $V^*_{i',j}$ in cell $c_{i',j}$, and can then only use a cell edge 
of~$c_{i'+1,j}$, a horizontal connector edge to reach vertical layer $j-1$, or a 
vertical connector edge to go back up to~$c_{i',j}$. In case~$i'=K$, the path 
reaches vertex $v_{bj}\in V_{(b)}$ from a vertex of $V^*_{i',j}$ of cell 
$c_{i',j}$, and since the path does not contain $t_2$ it can only use a 
vertical connector edge to go back up to another vertex of $V^*_{i',j}$ 
of~$c_{i',j}$.

Note that by \cref{claim:dis}, to reach a vertex of $V^*_{i',j}$ from $v_{il}$ 
before moving downwards, in some horizontal layer the path $P_i$ needs to use a 
connector edge between vertical layers $j-1$ and $j$, and it needs to use a cell 
edge of vertical layer $j$. Thus in the first case when $i'\leq K-1$ and $P_i$ 
uses a cell edge of $c_{i'+1,j}$, this path uses two cell edges of vertical 
layer~$j$. However, this contradicts property~(1) of \cref{lem:sslb}. In the 
second case, when $P_i$ uses a horizontal connector edge to reach vertical 
layer~$j-1$ after moving downwards, it uses two connector edges between 
vertical layers $j-1$ and $j$, which however contradicts property~(2) of 
\cref{lem:sslb}. The remaining case is when $i'\leq K$ and $P_i$ uses a vertical 
connector edge to go back up to~$c_{i',j}$. This implies that between horizontal 
layers $i'$ and $i'+1$ the solution contains two vertical connector edges, which 
are incident to the same vertex of horizontal layer $i'+1$. From property~(2) 
of \cref{lem:sslb}, these two edges must be used by two distinct vertical paths 
of $Q_1,\ldots,Q_K$, which however contradicts the fact that these paths are 
vertex disjoint.

Hence, $P_i$ cannot move downwards into a different horizontal layer. Using this 
observation, note that the horizontal paths induce a permutation $p$ on 
$\{1,2,\ldots, K\}$, where for $i,j \in [K]$ we have $p(i)=j$ if a path starting 
at the vertex $v_{li}$ on the left ends up at the vertex $v_{rj}$ on the right. 
If $p$ is not the identity function, then there must be indices $i<j$ such that 
$p(i)=j$. But this also means that the corresponding horizontal path moves down 
at some point to reach the layer $j$, which is a contradiction. Now that we know 
that every horizontal path starts and ends in the same layers, suppose there is 
a path that does not stick to its layer in between and takes a detour through 
some other layers. However, the path has to go downwards from a layer (not 
necessarily the one it started in) at some point, which again is a 
contradiction.

Through an analogous argument we can first argue that a vertical path from 
$Q_1,\ldots,Q_K$ cannot move \emph{rightwards}, and similarly if any path goes 
out of the corresponding vertical layer at any point we arrive at a 
contradiction as before.
\end{proof}

In conclusion, any solution of cost $3K^2 + 8K +3$ to the \VCSNDP instance 
implies a solution to the \pname{Grid Tiling} instance: by 
\cref{lem:sslb,lem:ss2}, each cell contains exactly one cell edge of the 
solution, which encodes an integer pair $(a,b)$ of the \pname{Grid Tiling} 
instance. Since the solution connects these cell edges by connector edges, the 
coordinates of the integer pairs agree accordingly. Thus together with 
\cref{lem:ssub} and the W[1]-hardness of the \pname{Grid Tiling} problem we 
obtain \cref{thm:VC-hard}.

\section{Hardness of 2-DST}\label{sec:2dst}

In this section we consider arguably one of the simplest versions of \SNDP in 
directed graphs: the \pname{$2$-Connected Directed Steiner Tree (\DST)} problem. 
We are given a directed graph $G = (V,E)$ with edge-costs $\cost(e)$ for $e\in 
E$, and a set of $\nrterm$ terminals $\terms \subseteq V$ with root $r \in 
\terms$. The goal is to compute a subgraph $H$ of $G$ with minimum cost, that 
has at least $2$ edge-disjoint paths from $r$ to each $t \in \terms \setminus 
\{r\}$. The \DST problem is a natural generalization of the classical 
\pname{Directed Steiner Tree (DST)} problem, where only one $r\to t$ path is 
required to exist for each~$t \in \terms \setminus \{r\}$. While DST is know to 
be FPT~\cite{dreyfus1971steiner} for parameter $\nrterm$, for the \DST problem 
we show the following:

\DSThard*

We show a reduction from the \pname{Multicoloured Clique} problem to the \DST 
problem. The input of \pname{Multicoloured Clique} consists of an undirected 
graph $G$, an integer $K$, and a partition $(V_1,\ldots ,V_K)$ of the vertices 
of $G$. The goal is to decide if there is a $K$-clique containing exactly one 
vertex from each set $V_i$. If such a clique exists the instance is a 
\emph{YES-instance} and otherwise a \emph{NO-instance} (i.e., it is a 
decision problem). \pname{Multicoloured Clique} is \textup{W[1]}-hard when 
parameterized by $K$ \cite{FellowsHRV09}. 

\subsection{Construction.}\label{subsec:dstconstruct}
Let $(G,K,(V_1,\ldots ,V_K))$ be an $n$-vertex instance of \pname{Multicoloured 
Clique}. For $i,j \in [K]$ where~$i <j$, let $E_{ij}$ be the set of 
(undirected) edges between $V_i$ and $V_j$. For any directed edge~$(u,v)$ we 
call $u$ the \emph{tail} and $v$ the \emph{head}. We construct an instance $G'$ 
of \DST from it as follows (cf.~\cref{fig:3}):

\begin{figure}
\centering
\includegraphics[scale=0.6]{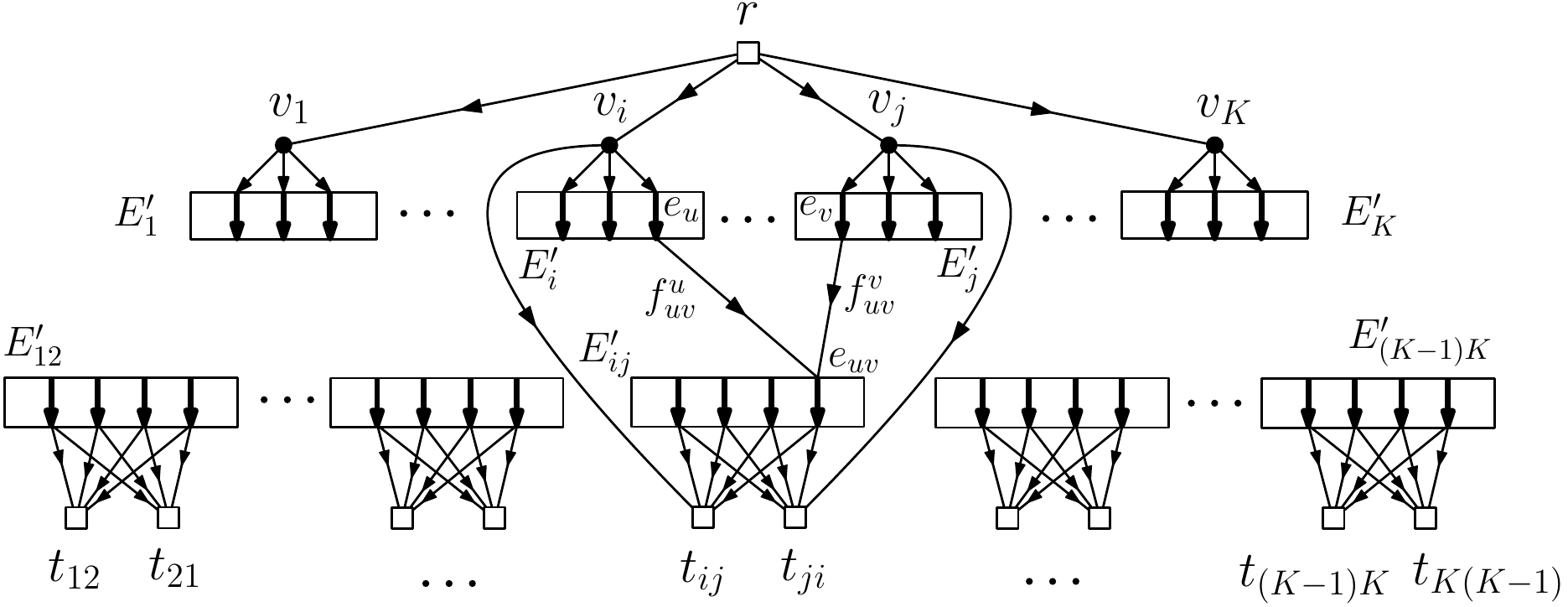}
\caption{Graph $G'$ with some edge connections shown for the root $r$ and the  
terminals $t_{ij}$ and $t_{ji}$. The thicker edges denote the edges of cost $1$ 
in the graph and the rest of the edges are of cost $0$.}
\label{fig:3}
\end{figure}

\begin{enumerate}
\item For each $i \in [K]$, create a directed edge set $E'_i$ with a distinct edge $e_v$ for each $v \in V_i$.
\item For $i,j \in [K]$ where $i <j$, create a directed edge set $E'_{ij}$ 
with an edge $e_{uv}$ for each (undirected) edge~$\{u,v\} \in E_{ij}$.
\item For $i,j \in [K]$ where $i <j$ and each edge $\{u,v\}\in E_{i,j}$, add 
two directed edges $f_{uv}^u$ and $f_{uv}^v$ from the heads of $e_u\in E'_i$ 
and~$e_v\in E'_j$, respectively, to the tail of $e_{uv}\in E'_{ij}$.
\item For $i,j \in [K]$ where $i < j$, create two vertices $t_{ij}$ and 
$t_{ji}$, and put directed edges from all tails of edges in~$E'_{ij}$ to both 
$t_{ij}$ and $t_{ji}$.
\item Add $K$ new vertices $v_1, \ldots, v_K$ and for each $i \in [K]$, add a 
directed edge from $v_i$ to all tails of edges in~$E'_i$ and also add an 
edge to $t_{ij}$ for each $j \in [K]$. 
\item Add a vertex $r$ and put edges from $r$ to the vertices $v_1, \ldots, 
v_K$.
\item The edges in $E'_i$ and $E' _{ij}$ are of cost $1$, and all other edges 
are of cost $0$.
\end{enumerate}

This completes the construction of the target \DST instance $G'$ where $r$ is 
the root vertex and the terminal set is $\terms = \{r, t_{ij}\mid i\neq j\}$. 
It is easy to see that the above construction can be done in polynomial time.

\subsection{Correctness.}
To show the correctness of the reduction, we argue that the \DST instance has a 
solution~$H$ with $\cost(H) \le K+{K \choose 2}$ if and only if the given 
\pname{Multicoloured Clique} instance is a YES-instance. In particular, first in 
\cref{lem:dstub} we show how to construct a solution of cost $K+{K \choose 2}$ 
given a YES-instance of \pname{Multicoloured Clique}. Next in \cref{lem:dstlb} 
we prove that any solution to the \DST instance has a cost at least $K+{K 
\choose 2}$. Finally, in \cref{lem:dstmain} we show that every cost $K+{K 
\choose 2}$ solution to the \DST instance must correspond to a solution to the 
\pname{Multicoloured Clique} instance. Since \cref{lem:dstub} also bounds the 
solution size, \cref{thm:2DST-hard} follows.

\begin{lemma}\label{lem:dstub}
Given a YES-instance of \pname{Multicoloured Clique}, the constructed \DST 
instance has a solution of cost $K+{K \choose 2}$. Furthermore, its solution 
size is $\solsize=O(K^2)$.
\end{lemma}
\begin{proof}
Since the given \pname{Multicoloured Clique} instance is a YES-instance, for each 
$1\le i < j \le K$, there is an edge $\{u,v\} \in E_{ij}$ for a unique vertex 
$u \in V_i$ and a unique vertex $v \in V_j$, that are part of a $K$-clique, say 
$C$. In the \DST instance we construct a solution from $C$ as follows: for each 
$\{u,v\} \in E_{ij}$ in $C$ add $e_u \in E'_i$, $e_v \in E'_j$, and $e_{uv} 
\in E'_{ij}$ to the solution and also add the two edges $f_{uv}^u$ and 
$f_{uv}^v$ connecting $e_u$ and $e_v$ to $e_{uv}$. Further, add the edges 
connecting each $v_i$ with the corresponding edge~$e_v$ where $v\in V_i$, the 
two edges connecting $e_{uv}$ to~$t_{ij}, t_{ji}$, and also the two edges $(v_i, 
t_{ij})$ and~$(v_j, t_{ji})$. Finally, to complete the solution add the $K$ 
edges $(r, v_i)$ for each~$i \in [K]$. 

This indeed gives a solution to the \DST instance, since for each terminal 
$t_{ij}$ we have two edge-disjoint paths from $r$: a length-two path via $v_i$ 
and a path from $v_j$ to some $e_{uv} \in E'_{ij}$ via~$e_v \in E'_{j}$, which 
must exist as the given \pname{Multicoloured Clique} instance is a 
YES-instance. Since there is a unique vertex $u$ in each $V_i$ associated with a 
$K$-clique $C$, we add to our solution in the \DST instance exactly the $K$ 
edges associated with those vertices, one $e_u$ in each $E'_i$. Also, we add to 
the solution the ${K \choose 2}$ edges, corresponding to the edges in the 
$K$-clique $C$, one~$e_{uv}$ in each $E'_{ij}$. Hence, the cost of our solution 
is exactly $K+{K \choose 2}$ as all other edges added are of cost $0$.

From our construction, the total number of remaining edges in the solution is 
bounded by $2K+ K(K-1) + 2K(K-1) = O(K^2)$. Here the first term corresponds to 
the $K$ length-two paths connecting $r$ to each $E'_i$ via $v_i$, for $i \in 
[K]$. The second term is due the edges from $v_i$ to~$t_{ij}$, one for each of 
the $K(K-1)$ terminals. For each $e_{uv} \in E'_{ij}$ in the solution, the final 
term correspond to the two edges from $E'_i$ and $E'_j$ to $e_{uv}$ and the two 
edges from $e_{uv}$ to $t_{ij}$ and $t_{ji}$. Hence, we obtain~$\solsize=O(K^2)$, 
as required.
\end{proof}

We remark that the following lemma crucially uses the directedness of $G'$ in 
order to control the routes of the paths. For undirected graphs such a (rather 
straightforward) reduction seems hard to achieve, as mentioned in 
\cref{sec:techniques}.

\begin{lemma}\label{lem:dstlb}
Any solution to the \DST instance has a cost of at least $K+{K \choose 2}$. 
Moreover, in any solution of cost~$K+{K \choose 2}$ there is exactly one edge 
picked from each $E'_i$ and each $E'_{ij}$, and it contains some edge $f^u_{uv}$ 
(resp.,~$f^v_{uv}$) connecting the selected edge $e_u\in E'_i$ (resp., $e_v\in 
E'_j$) to the selected edge~$e_{uv}\in E'_{ij}$.
\end{lemma}
\begin{proof}
To construct a solution to the \DST instance, for each terminal $t_{ij}$ we 
must pick two edge-disjoint paths from $r$ to $t_{ij}$. Notice that all the 
paths from $r$ to $t_{ij}$ go via $v_i$ or $v_j$, since $t_{ij}$ can only be 
reached from $v_i$, $E'_{ij}$, $E'_i$, and $E'_j$, in addition to $v_j$ and 
$r$. As there is only one edge from $r$ to each $v_i$, due to the 
edge-disjointness condition we must pick one path through $v_i$ and another 
through $v_j$ to reach $t_{ij}$. Though it is possible to use the 
$(v_i,t_{ij})$ edge, which is of cost $0$, any path connecting $v_j$ to $t_{ij}$ 
must use an edge from $E'_j$ and $E'_{ij}$, which have cost~$1$ each. Similarly 
for $t_{ji}$ the paths must go via both $v_i, v_j$ and one of the paths goes via 
$E'_i$ and $E'_{ij}$. Thus together the terminals $t_{ij}, t_{ji}$ ensure that 
any solution uses at least one edge from each of $E'_i$, $E'_j$, and $E'_{ij}$, 
together with some edges $f_{uv}^u$ and $f_{uv}^v$ connecting these. Hence any 
solution must incur a cost of at least~$K+{K \choose 2}$. Moreover, any solution 
of cost $K+{K \choose 2}$ can only afford one edge from each of the edges sets 
$E'_i$ and~$E'_{ij}$.
\end{proof}
\begin{lemma}\label{lem:dstmain}
Any solution of cost $K+{K \choose 2}$ for the \DST instance implies a solution 
to the \pname{Multicoloured Clique} instance.
\end{lemma}
\begin{proof}
From \cref{lem:dstlb} we know that in any solution to the \DST instance of cost 
at most~$K+{K \choose 2}$, for any indices $i<j$ there is a unique edge $e_u$ 
used in $E'_i$, a unique edge $e_v$ used in $E'_j$, and a unique edge $e_{uv}$ 
used in~$E'_{ij}$. 
This means that the solution must also contain edges $f_{uv}^u$ and $f_{uv}^v$ 
connecting $E'_i$ and $E'_j$ to $E'_{ij}$. We can now find a clique in the 
\pname{Multicoloured Clique} instance by selecting the vertices $u\in V_i$ and
$v\in V_j$ for each $i,j\in[K]$ where $i<j$. The existence of the connecting 
edges $f_{uv}^u$ and $f_{uv}^v$ imply the existence of an edge $\{u,v\}$ 
between any pair of selected vertices, i.e., the vertices induce a clique.
\end{proof}

\section{Open problems}\label{sec:open}

While we settle the parameterized complexity of quite a few cases of \SNDP in 
this paper, we leave some open questions as well:

\begin{itemize}
\item Given that \VCSNDP is W[1]-hard parameterized by the solution size, an 
obvious question becomes: is it possible to approximate this problem in FPT time 
within a factor better than the known polynomial-time approximation algorithms, 
possibly even beating the known approximation lower bounds? The same question 
can be posed for~$\dem$-DST.

\item Our reduction for \VCSNDP excludes FPT algorithms parameterized by the 
maximum demand $\maxdem$, even for a constant number $\nrterm$ of 
terminals. At the same time, the reduction in~\cite{kortsarz2004hardness} 
together with the results of~\cite{dinur2018eth,manurangsi2021strongish} show 
hardness for the parameterization by $\nrterm$, but with unbounded~$\maxdem$. 
Note that these results do not exclude an algorithm with runtime 
$f(\nrterm)\cdot n^{g(\maxdem)}$ for some functions $f$ and $g$. Such an 
algorithm would nicely bridge the gap to the FPT 
algorithm~\cite{dreyfus1971steiner} for \pname{Steiner Forest}, where 
$\maxdem=1$ and $\nrterm$ is the parameter. In particular, can \VCSNDP be solved 
in $2^{O(\nrterm)}\cdot n^{O(\maxdem)}$ time?

\item We showed that \ECSNDP and \LCSNDP are FPT parameterized by the solution 
size $\solsize$ and gave algorithms with runtime $2^{O(\solsize)}\polyn$. What 
is the parameterized complexity of these problems for stronger parameters, such 
as the sum of demands $\sumdem$ or even the number $\nrterm$ of terminals? 

\item When considering input graphs of bounded treewidth, we gave an FPT 
algorithm for \VCSNDP parameterized by the treewidth $\tw$ and the sum of 
demands $\sumdem$. However, what is the parameterized complexity of this problem 
when combining the treewidth $\tw$ with the stronger parameter given by the 
number $\nrterm$ of terminals?

\item \citet{bateni2011approximation} give an approximation scheme for 
\pname{Steiner Forest} with XP runtime parameterized by the treewidth alone. 
They leave open whether this can be improved to FPT runtime. Considering \SNDP, 
is there an approximation scheme parameterized by only the treewidth $\tw$ with 
either XP or FPT runtime?
\end{itemize}

\medskip
\paragraph*{Acknowledgments.} We would like to thank Sasha Sami for pointing 
out a mistake in the conference version of this paper 
(cf.~footnote on \cpageref{footnote:Sasha}).

\printbibliography

\end{document}